\documentclass[a4paper]{article}

\usepackage[]{geometry}
\usepackage{times}  
\usepackage{helvet} 
\usepackage{courier}  
\usepackage[hyphens]{url}  
\usepackage{graphicx} 
\usepackage{natbib}  
\usepackage{caption} 
\usepackage{geometry}
\usepackage{times}
\usepackage{soul}
\usepackage{url}
\usepackage[hidelinks]{hyperref}
\usepackage[utf8]{inputenc}
\usepackage{caption}
\usepackage{graphicx}
\usepackage{amsmath}
\usepackage{amsthm}
\usepackage{booktabs}
\usepackage[ruled,vlined,linesnumbered]{algorithm2e}
\usepackage{enumerate}
\usepackage{amssymb}
\usepackage{optidef}
\usepackage{subfigure}
\usepackage{contour}
\usepackage{color}

\usepackage{appendix}

\newtheorem{definition}{Definition}
\newtheorem{lemma}{Lemma}

\newtheorem{theorem}{Theorem}
\newtheorem{corollary}{Corollary }

\newcommand{\RNum}[1]{\lowercase\expandafter{\romannumeral #1\relax}}
\newcommand{\RNumUpper}[1]{\uppercase\expandafter{\romannumeral #1\relax}}

\usepackage{lipsum}

\newcommand\blfootnote[1]{%
  \begingroup
  \renewcommand\thefootnote{}\footnote{#1}%
  \addtocounter{footnote}{-1}%
  \endgroup
}

\title{Cooperation, Retaliation and Forgiveness in Revision Games\blfootnote{Emails: haodong@uestc.edu.cn, qishi@std.uestc.edu.cn; huizhang@std.uestc.edu.cn; boan@ntu.edu.sg}}
\author{Dong Hao$^1$\blfootnote{Dong Hao is the corresponding author.}\blfootnote{This work has been submitted to the IEEE for possible publication. Copyright may be transferred without notice, after which this version may no longer be accessible.} \and Qi Shi$^1$ \and Jinyan Su$^1$ \and Bo An$^2$}
\date{%
    $^1$University of Electronic Science and Technology of China\\%
    $^2$Nanyang Technological University\\[2ex]%
}

\begin{document}

\maketitle
	
\begin{abstract}
Revision game is a very new model formulating the  real-time situation where players dynamically prepare and revise their actions in advance before a deadline when payoffs are realized. It is at the cutting edge of dynamic game theory and can be applied in many real-world scenarios, such as eBay auction, stock market, election, online games, crowdsourcing, etc. In this work, we novelly identify a class of strategies for revision games which are called Limited Retaliation strategies. An limited retaliation strategy stipulates that, (\RNum 1) players first follow a recommended cooperative plan; (\RNum 2) if anyone deviates from the plan, the limited retaliation player retaliates by using the defection action for a limited duration; (\RNum 3) after the retaliation, the limited retaliation player returns to the cooperative plan. A limited retaliation strategy has three key features. It is cooperative, sustaining a high level of social welfare. It is vengeful, deterring the opponent from betrayal by threatening with a future retaliation. It is yet forgiving, since it resumes cooperation after a proper retaliation. The cooperativeness and vengefulness make it constitute cooperative subgame perfect equilibrium, while the forgiveness makes it tolerate occasional mistakes. limited retaliation strategies show significant advantages over Grim Trigger, which is currently the only known strategy for revision games.
Besides its contribution as a new robust and welfare-optimizing equilibrium strategy, our results about limited retaliation strategy can also be used to explain how easy cooperation can happen, and why forgiveness emerges in real-world multi-agent interactions. In the application aspect, limited retaliation strategies are especially suitable for real-world games with deadlines where players confront observation or action errors, or where players have limited rationality. In addition, limited retaliation strategies are simple to derive and computationally efficient, making it easy for algorithm design and implementation in many multi-agent systems.

\end{abstract}

\section{Introduction}
\subsection{Background and Related Works}
In many multi-agent systems, agents often prepare their actions in advance before a deadline when their payoffs are realized. Such systems can be captured by revision game, which is an emerging model of dynamic game theory \cite{kamada2020revision,kamada2020cooperation}. In a revision game, at the initial time, agents choose their actions and wait for the deadline. Between the initial time and the deadline, they may have opportunities to revise actions. The revision opportunities arrive stochastically and the probability to get another revision opportunity vanishes as time approaches the deadline. A player's only payoff is obtained at the deadline, depending on all agents' final actions.

Many real-world scenarios can be modeled as revision games. A well-studied case is the pre-opening phase in stock markets such as Nasdaq or Euronext \cite{cao2000price,biais1999price}. Traders can submit orders before the opening of the market, which can be changed until the opening time. The online auction websites like eBay hold auctions which have a deadline and bidders actually play a revision game \cite{roth2002last}. The eBay auctions usually have a deadline, before which bidders can revise their bids many times. Bidders' opportunities to use eBay and to change bids are following a stochastic process (many human activities can be characterized by a Poisson process). 
Their revising the strategies must take into consideration of the bids of the opponents and the strategy will not bring about revenue until the auction ends. Similar auctions include spectrum auction whereby a government uses an auction system to sell the rights to transmit signals over specific bands of the electromagnetic spectrum and to assign scarce spectrum resources, and such big auctions usually have fixed deadlines \cite{milgrom2021auction}.
Electoral campaign is also a good example \cite{warshaw2020fatalities}. During the $2020$ U.S. presidential election, the COVID$-19$ pandemic brings about a series of political agenda shifts and campaign tactic adjustments for the major presidential candidates.
AI game playing problems can also be modeled as revision games. Many recent multiplayer video games (such as \emph{Need for Speed} series and \emph{Call of Duty} series) involve some form of time limit on their online servers \cite{claypool2006latency,claypool2010latency}. Whichever team/player holds the lead at the expiration of the time limit is declared the winner. Another example is federated learning. Large companies may ask individuals to submit their data to train a model before some fixed deadline \cite{prakash2020coded}. The relation between different individuals can be modeled as a Prisoner's Dilemma. Moreover, crowdsourcing competitions or contests on Github or TopCoder usually have a deadline; participants modify their submissions many times before the deadline \cite{saremi2020right}.

All above mentioned secnarios can be modeled as a game where players prepare and revise their strategies before a deadline, which is the model of the revision game \cite{kamada2020revision}. A revision game starts at time $-T$ and ends at time $0$. Players prepare actions at he beginning of the game and they can only revise their actions when they obtain revision opportunity, which arises according to a Poisson process and the payers simultaneously revise their actions. The actions that are played at the ``deadline" determined their final payoff.
The general framework of revision games is very new and is formally published in \cite{kamada2020revision,kamada2020cooperation}. But the original idea can be traced back to the authors' online working draft before 2014. They identify the very first strategy for revision games, which is based on the grim trigger mechanism, and constitutes equilibrium.
	Following Kamada and Kanrodi, all the existing strategies of revision games follow a grim trigger style. Strategies of this style has a very fierce punishment and is never forgiving; it only proves cooperation is possible but does not answer how easy it is in revision games. It also excludes the possibility of resuming cooperation. A few works study the grim trigger equilibrium in games where the revision process is player-specific  \cite{calcagno2014asynchronicity},\cite{kamada2014valence}, \cite{gensbittel2018zero}. A later work examines the trembling hand equilibrium for revision games with imperfect information \cite{moroni2015existence}. Another work goes even further and extends the solution concept to Markov perfect equilibrium \cite{lovo2015markov}. Besides these theoretical achievements, researches also investigate different applications \cite{ambrus2015continuous, kamada2020cooperation} and even conduct laboratory experiments for revision games \cite{roy2017revision}.

\subsection{Intuition and Contribution}


	Our problem is about how to sustain mutual cooperation and how to improve the social welfare in revision games. Revision game are dynamic. It lasts for a certain duration (i.e., time interval $[-T,0]$). When players are engaged in a revision game, mutual cooperation could be achieved while players actions evolve. This is due to the fact that agents must consider not only the immediate best response but also the long-term expected payoffs. Thus one player has the ability to deter the opponent from non-cooperative immediate action by threatening retaliation that reduces the opponent's expected payoff. 
	
	In revision games, the cooperation and mutual threat can be realized by an equilibrium strategy, which is consisted of two parts: cooperative plan and punishment scheme. As long as the opponent deviates from the cooperative plan at any time $t$, a player will start a punishment scheme. For example, for a player $i$, she can cooperate in the beginning, if she see her opponent $j$ deviates from cooperating at any time $-t$, then $i$ will retaliate against $j$ in the remaining time. Such a potential retaliation is a threatening that reduces $j$'s long-term expected payoff, whereby $i$ can deter $j$ from exploiting his short-term advantage. In this way, a certain level of mutual cooperation can be sustained. If a profile of such an equilibrium strategy is a mutual best response, then players playing such strategies is an equilibrium for the revision game.

In this work, we follow the standard model of the revision games given by \cite{kamada2020revision}. However, we novelly identify the Limited Retaliation (LR) strategies for revision games. LR strategies consist of two components: One is the recommended \emph{\textbf{collusive plan}} which achieves an optimal social welfare for players; The other is the \emph{\textbf{retaliation phase}} in which the defection action is used for a limited period. The LR strategy stipulates that, if anyone fails to follow the collusive plan, players switch to the retaliation phase. LR strategies work in a ``carrot-and-stick'' style, where the collusive plan is the carrot and retaliation phase is the stick.
Specifically, our contributions are:

    \begin{enumerate}
        \item   We introduce the one-shot deviation principle into revision games and derive the necessary and sufficient condition for LR strategies to constitute subgame perfect equilibrium. The conditions are captured by an incentive constraint which is the difference between the gain for current deviation and the loss for future retaliation.

        \item  By satisfying the incentive constraint using backward induction, We find a class of equilibrium LR strategies that sustain a high level of mutual cooperation. Their plans are given in a simple form and are easy to implement. We also obtain the social welfare maximizing ones in this class of LR strategies. These collusive plans start from fully cooperative action, evolves over time and become more and more in-cooperative as time elapses.

        \item  The retaliation phase of LR strategies is designed to have a limited and controllable fierceness. It not only can exert enough threat on the opponent to ensure him not to deviate, but also allows players to forgive the deviator once his deviation has been sufficiently retaliated. Thus players can resume collusion after this retaliation phase.  The forgiveness makes LR strategies especially robust against noises, errors, and limited rationality.

        \item  We also implement the optimal LR strategies in both the revision versions of Prisoner's Dilemma and the Cournot games. Empirical results show that LR strategies show a significant robustness than grim trigger strategies.
    \end{enumerate}

Our LR strategies is the first attempt to explain \textit{how easy cooperation can happen}, and \textit{why forgiveness emerges} in multi-agent interactions with a deadline.
The LR strategies have three good features: cooperative, vengeful and forgiving. They are cooperative since LR strategy players follow a coordination plan if there is not deviation. Cooperativeness makes the payers enjoy high social welfare \cite{orbell1993social}. LR strategies are vengeful since they can deter the opponent from non-cooperative action at any time by threatening a sufficient retaliation that reduces the opponent's expected payoff in the remaining time. By using this vengeful feature, LR strategies can constitute cooperative subgame perfect equilibrium \cite{bielefeld1988reexamination}. The other feature is they are forgiving. Since the retaliation phase has a controllable length, even when there is an intentional or unintentional deviation, LR players can return back to mutual cooperation after the retaliation. This forgiving feature makes LR strategies especially suitable for real-world games where players confront observation or action errors \cite{gilpin2007lossless,fundenberg1990evolution}, or where players have limited rationality \cite{doyle1992rationality}.

LR strategies are in contrast to existing strategies for revision games which assumes players are extremely grim and never forgive. The forgiveness in LR strategies is of great importance. This is because in realistic games, due to agents' observation/action errors or bounded rationality, deviation occasionally happens either intentionally or unintentionally. If players play a non-forgiving strategy, they can never return back to mutual cooperation and the social welfare would be significantly harmed. Moreover, forgiveness widely exists in the real world, our LR strategy is the first attempt to explain the incentive and mechanism for cooperation and forgiveness in multi-agent games with deadline.

\section{Revision Game Model}
\subsubsection{Game Setting}
Consider a general two-player symmetric game in which there is a  Nash equilibrium action profile $(a^N,a^N)$, where both players defect, and the equilibrium payoffs $\pi^N:=\pi(a^N,a^N)$. For such a game, denote the optimal symmetric action profile as $(a^*,a^*)$ and $a^*:=\arg \max_{a} \pi(a,a)$. The optimal action profile is the social welfare maximizing, in the sense that the system's total payoff is maximized. The possible action space $A$ for each player is a convex subset (interval) in $\mathbb R$.
Assume any player $i$'s payoff $\pi_i$ is continuous.
    Denote the payoff for a symmetric action profile  $(a,a)$ as $\pi(a)$. Assume this payoff function is strictly increasing for $a \in[a^N, a^*]$ and non-decreasing for $a>a^*$, if $a^N < a^*$ (symmetric conditions hold if $a^N > a^*$). 
The conventional Prisoner's Dilemma with continuous actions and Cournot duopoly both have these two properties. When such a game happens only once, we call it the \emph{{stage game}}. 


A revision game is a dynamic extension of the above stage game. 
It starts at time $-T$ ($T>0$) and ends at deadline $0$. At $-T$, players choose actions simultaneously. At any time $-t\in \left(T,0\right]$, they may have \textit{\textbf{revision opportunities}} to play the stage game again, and the opportunities arrive according to a Poisson process with arrival rate $\lambda>0$. Players can revise their actions only when they have a revision opportunity. 
In Revision games, players' payoffs $\pi_i(a'_i,a'_j)$ and $\pi_j(a'_i,a'_j)$ are realized at the deadline, where $(a'_i,a'_j)$ is the profile of players' revised actions at the final revision opportunity. After the final revision opportunity, their actions stay the same till the deadline.

Players' behaviors in the revision game are captured by a mapping $\sigma:[-T,0]\rightarrow A$. We  call this mapping $\sigma(\cdot)$ the \emph{{strategy}} of  players. At each  time point $-t$, the strategy decides an action $\sigma(t)$. The value $t$ is the remaining \textit{time length} until the deadline, which can be regarded as a payoff-relevant state variable for the players. Note that when there is a revision opportunity at $-t$, a player may or may not adjust her action. 
We say that an action $\sigma(t)$ is \emph{\textbf{cooperative}} or $\sigma(t)$ sustains a certain level of \emph{{cooperation}}, if it provides a \emph{higher payoff than the Nash equilibrium payoff of the stage game}. That is, if both players use $\sigma(t)$, then $\pi\big(\sigma(t)\big)=\pi_i\big(\sigma(t),\sigma(t)\big)>\pi^N$. One can see there are different levels of cooperation.  
	
\subsubsection{Equilibirum Concept}
Now we introduce the revision opportunities, quantify the expected payoffs, and define the equilibrium concept we will use in this paper. 

With the Poisson arrival rate $\lambda$, we know $e^{-\lambda t}$ is the probability of no Poisson arrival in the remaining time length $t$ and $\lambda e^{-\lambda t}$ is the probability density of the last revision opportunity.
The expected payoff to each player associated with strategy $\sigma$ over time span $[-T,0]$ is:
\begin{equation}\label{Eq:expected_payoff}\begin{aligned}
V(\sigma, T)=\pi\big(\sigma(T)\big)e^{-\lambda T}+\int_0^T\pi\big(\sigma(t)\big)\lambda e^{-\lambda t} \mathrm{d}t.
\end{aligned}\end{equation}
Eq.(\ref{Eq:expected_payoff}) is constructed as follows: $\sigma(T)$ specifies players' initial actions. Payoff $\pi\big(\sigma(T)\big)$ can be realized only when there is no revision opportunities arising during $[-T,0]$, which happens with probability $e^{-\lambda T}$. When a revision opportunity arrives
at time $-t$, players choose action $\sigma(t)$. Payoff $\pi\big(\sigma(t)\big)$ is realized only when there is no revision opportunities arrive after $-t$, and this happens with a probability density $\lambda e^{-\lambda t}$.
Thus the integral is the expected payoff when the last revision opportunity arrives during $(-T,0]$. 


In a revision game, a {\textit{subgame}} $\Lambda(t)$ begins at a particular time point  $-t$, and contains the rest of the game from that point forward.
The whole revision game is the largest subgame, and the smallest subgame is the single stage game at the deadline. Every subgame can be transformed into a normal-form bi-matrix game. 
 For the smallest subgame, the action profile $(a^N,a^N)$ is its only Nash equilibrium. However, for other subgames which can be represented by a large matrix, there could be multiple Nash equilibria. 
 \begin{definition}[NE]\label{NE Definition}
A strategy profile ${{\bf \sigma}=(\sigma_i, \sigma_{-i})}$ is a Nash equilibrium at a given time $-t$, if for any player $i$, all $\sigma_i'\ne \sigma_i$, the expected payoff for the subgame $\Lambda(t)$ satisfies
$V_i\left((\sigma_i, \sigma_{-i}), t\right)\ge V_i\left((\sigma_i', \sigma_{-i}), t\right)$.
\end{definition}
 
We focus on the refinement of Nash equilibrium, which is called \textit{subgame perfect equilibrium (SPE)}\cite{bielefeld1988reexamination}.
\begin{definition}[SPE]\label{SPE Definition}
A strategy profile ${{\bf \sigma}=(\sigma_i, \sigma_{-i})}$ is a subgame perfect equilibrium for the whole revision game, if for any player $i$, all $\sigma_i'\ne \sigma_i$,
$V_i\left((\sigma_i, \sigma_{-i}), t\right)\ge V_i\left((\sigma_i', \sigma_{-i}), t\right)$ for all subgames $\Lambda(t)$ with any $t\in[-T,0]$.
\end{definition}
That is, a strategy profile ${\bf \sigma}$ is a subgame perfect equilibrium if it represents a Nash equilibrium of every subgame of the whole revision game. 
Every SPE is a NE, but not the other way around. These two definitions also cover the  revision games with symmetric actions and payoffs.
In the following contents, we use SPE as our solution concept.

	\section{Limited Retaliation Strategies}
	In this section, we propose Limited Retaliation (LR) strategies that constitute SPE. They are the only known equilibrium strategies in revision games besides Grim Trigger.

\subsection{Formalization of LR Strategies}

\subsubsection{Limited Retaliation Strategy}The intuition for designing LR strategies is: Firstly, it should recommend a cooperative plan. If both players follow this plan, they will have high expected payoffs. Secondly, to ensure players not to deviate from the cooperative plan, the strategy should also provide a  retaliation which is used to threaten the opponent not to deviate. Finally, if retaliation is too soft, it cannot have enough threat; if it is too strong, it will be hard to return to cooperation. Thus the retaliation should have a limited degree.
 	\begin{definition}[LR Strategy]\label{Definition:LR_strategy}
     A Limited Retaliation strategy is a tuple $\sigma=\{x(\cdot),k\}$. Function $x(\cdot)$ is a recommended cooperative plan specifying actions $a=x(t)$ for each $-t$. Coefficient $k\in(0,1)$ scales the duration of the retaliation phase, where both players use the defection action $a^N$.
	\end{definition}

LR Players start with the initial action $x(T)$. When a revision opportunity arrives at any time $-t$, they choose action specified by $x(t)$. If
any player deviates from $x(t)$ when a revision opportunity arrives at $-t$, both players
choose the Nash defection action $a^N$ in all revision opportunities
arriving during period $(-t,-t+kt]$; if no revision opportunities arrive during this period, don't need to do anything. After this period, when a revision opportunity arrives, return back to follow  $x(\cdot)$ no matter what happened during $(-t,-t+kt]$.

We would like to emphasize that the cooperative plan $x(\cdot)$ should not be confused with the actual path of revised actions. The intuition of an LR strategy $\sigma=\{x(\cdot),k\}$ is: it wants both players' actual action paths to agree with the  cooperative plan $x(\cdot)$. This is ensured by threatening the opponent with a retaliation quantified by coefficient $k$.


\subsubsection{Cooperative Plan}We now clarify what is a \textit{cooperative} plan. Recall that in stage games (e.g., Prisoner's Dilemma), the defection action $a^N$ is the Nash equilibrium action, and it leads to the lowest social welfare. A fundamental objective in revision games is to incentivize the players to avoid taking this defection action $a^N$ for a long enough time. With this objective, the cooperative plan is defined as follows.
\begin{definition}[Cooperative Plan]
     A plan $x(\cdot)$ recommended by the LR strategy is called cooperative, iff $\forall-t, x(t)\ge a^N$. That is, at each time $-t$, $x(\cdot)$ specifies a cooperative action, which leads to better social welfare than $a^N$.
\end{definition}

By definition of LR strategies, as long as any player (either the opponent or the focal player, or both) deviates at $-t$, then the focal player will change to $a^N$ for a limited period with coefficient $k$. The deviator's retaliation (i.e., changing to $a^N$) after her own deviation can be seen as a \textit{self-defense} since the focal player knows that after her deviation, the opponent will start retaliation against her, and the Nash action $a^N$ is   the \textit{secure} action for stage games. 



If the retaliation is never triggered, then both players' actual path of play perfectly follow the recommended cooperative plan $x$. This indicates that with LR strategies, \emph{{cooperation can emerge}}. If there is any deviation by any player, then the retaliation phase is triggered. This potential retaliation phase is a threatening to the opponent not to deviate, indicating that \emph{{cooperation can be secured}}. LR strategies are ``forgiving'' and the fierceness of retaliation is controlled by $k$. After the retaliation, LR strategies return to cooperation. This indicates \emph{{cooperation can be restored}}.

\subsection{Algorithm for Implementing LR Strategies}\label{SI:section:Algorithm_for_Generating_LR}
The following algorithm can be used to implementing  an LR strategy. To run this algorithm, one needs to decide the collusive plan $x$, retaliation fierceness $k$ in advance. These issues have been thoroughly discussed in the main text sections ``Equilibrium Computation'' and ``Social Welfare Maximization''.

 	\begin{algorithm}[ht]
		\KwIn{collusive plan $x(\cdot)$, defection action $a^N$, \\retaliation degree $k$, initial time $-T$}
		\KwOut{LR strategy $\sigma_i(t)$ for all $t\in[0,T]$}

		\textbf{Initialize}: \textbf{set} $\sigma(T)=x(T)$\\

        \For{$-T <-t < 0$}{
            \If{$\sigma_i(t-1)=\sigma_j(t-1)=x(t-1)$}{
		        \If{revision opportunity arrives at $-t$}{
		            \textbf{set} $\sigma_i(t)=x(t)$
		        }
		        \Else{
		            \textbf{set} $\sigma_i(t)=\sigma_i(t-1)$
		        }
		        $-t=-t+1$
		    }
		    \Else{		
			    \textbf{set} $-\tau=-t+kt$ \\
			    \While {$-t\le-\tau$ }{
 		            \If{revision opportunity arrives at $-t$}{
 		                 $\sigma_i(t)=a^N$ 
 		            }
 		            \Else{
 		                 $\sigma_i(t)=\sigma_i(t-1)$
 		            }
 		        $-t=-t+1$
 	       	    }
 		    }
        }
		\caption{LR strategy for an arbitrary player $i$}
		\label{Alg:LR-K}
	\end{algorithm}

	In this algorithm, $\sigma_i(t)$ (or $\sigma_j(t)$) is player $i$'s (or $j$'s) action specified by the LR strategy at time $-t$. Function $x$ is the recommended collusive plan specifying an action at each $-t$. Line 3-8 shows if both players followed the recommended action at time $-t+1$, then the LR strategy player still follow the collusive plan at time $-t$. Otherwise, as it is shown in line 9-16, the LR player will set a retaliation length, and then use the defection action for this length. In addition, any adjustment of $\sigma_i(t)$ requires there should arrive a revision opportunity at time $-t$, which is captured by line 6-7 and line 14-15.

\section{Equilibrium Condition}
After the above formalization of LR strategies, now a natural question is: can LR strategies $\sigma=\{x(\cdot),k\}$ constitute SPE? This section is devoted to answering this question. We find a necessary and sufficient condition for any LR strategy to constitute (symmetric) subgame perfect equilibrium (SPE).


  
\subsection{One-Shot Deviation Principle}  
We first introduce the one-shot deviation principle (OSDP) tailored for revision games, by which we can further prove that LR strategies constitute SPE in the next subsection. 
  
OSDP is the key to verify whether a strategy ${\bf  \sigma}_i$ constitutes SPE \cite{mailath2006repeated}. A \textit{\textbf{one-shot deviation}} from  $\sigma_i$
is a strategy $\sigma'_i$ that agrees with $\sigma_i$ at all times except $-t$, i.e., $\exists ! t$ such that $\forall t^1\ne t, \sigma'_i(t^1)=\sigma_i(t^1)$. A one-shot deviation is \textit{profitable} if $    V_i\big((\sigma_i',\sigma_{-i}),t\big)>V_i\big((\sigma_i,\sigma_{-i}),t\big).$
We now tailor OSDP to revision games.
\begin{lemma}[One-Shot-Deviation Principle]\label{lemma:one-shot deviation principle}
In revision games, a strategy profile $\bf  \sigma$ is SPE iff there is no profitable one-shot deviation at any time.
\end{lemma}
\begin{proof}
 The ``only if'' condition is immediate from the definition of SPE. For sufficiency, suppose that for strategy profile ${\bf  \sigma}=(\sigma_i,\sigma_{-i})$, there is no profitable one-shot deviation, but it is not SPE. 
 Since $\bf  \sigma$ is not SPE, there must be some time $-t$ and a player $i$ with a better strategy $\sigma_i'\ne\sigma_i$, such that
\begin{equation}\label{eq:osd-1}
    V_i\big((\sigma_i',\sigma_{-i}),t\big)>V_i\big((\sigma_i,\sigma_{-i}),t\big).
\end{equation}
Let $-\bar \tau$ be the last time that $\sigma_i'$ and $\sigma_i$ differ.
Thus after time $-\bar \tau$, $\sigma_i'$ and $\sigma_i$
are identical. Since it is assumed there is no profitable one-shot deviation, at time $-\bar \tau$, $\sigma_i'$ cannot be profitable than $\sigma_i$, meaning that
\begin{equation}\label{eq:osd-2}
    V_i((\sigma_i,\sigma_{-i}),\bar \tau)\ge V_i((\sigma_i',\sigma_{-i}),\bar \tau)
\end{equation}
Let $\sigma_i''$ be a strategy equivalent to $\sigma_i'$ during $[-t,-\bar \tau -1)]$ and equals $\sigma_i$ at $-\bar\tau$. Therefore, by Eq.(\ref{eq:osd-2}) $\sigma_i''$ is better than $\sigma_i'$ at $-t$. Introducing this result into Eq.(\ref{eq:osd-1}) yields
\begin{equation}\label{eq:osd-3}
    V_i((\sigma_i'',\sigma_{-i}),t)
\ge V_i((\sigma_i',\sigma_{-i}),t)
\ge V_i((\sigma_i,\sigma_{-i}),t)\notag
\end{equation}
Here $\sigma_i''$ differs from $\sigma_i$ only up to $-\bar \tau-1$. Repeating the process, we obtain $\sigma_i'''$ satisfying 
$$V_i((\sigma_i''',\sigma_{-i}),t)
\ge V_i((\sigma_i,\sigma_{-i}),t)
$$
that differs from $\sigma_i$ only up to $-\bar\tau-2$. Thus, this process eventually yields a strategy $\tilde \sigma_i$ with 
$$V_i((\tilde \sigma_i,\sigma_{-i}), t)>V_i(( \sigma_i,\sigma_{-i}), t)$$
that only differs from $\sigma_i$ at time $t$. This contradicts the premise that  there is no profitable one-shot deviation. 
\end{proof} 

Lemma \ref{lemma:one-shot deviation principle} shows, if there is no profitable one-shot
deviation, then there is no other change in strategy (even very complicated changes or changes at many time points) that can increase payoffs.
This eliminates the hassle of what could happen after time $-t$. OSDP will be our key technique to prove that LR strategies constitute SPE.


\subsection{Subgame Perfect Equilibrium Constraint}
According to Lemma \ref{lemma:one-shot deviation principle}, checking whether an LR strategy can constitute SPE is equivalent to checking whether the LR players have a profitable one-shot deviation from the cooperative plan.\footnote{Regarding deviation, this paper only considers a change on the cooperative plan $x(\cdot)$, but not on the retaliation severeness $k$.}
By definition of LR strategy, a one-shot deviation will take advantage of the opponent and reap an immediate payoff increase. After that, this deviation will be retaliated for a period during which the deviator's payoff decreases. Therefore, the problem of finding SPE condition for LR strategies boils down to comparing the payoff increase and payoff decrease caused by  one-shot deviation at \textit{any} $-t$.

Technically, at any $-t$, if $j$ follows $x(t)$ but $i$ deviates to her payoff maximizing action $\arg\max \nolimits_{a_i} \pi_i(a_i,a)$, then $i$ has an immediate deviation gain quantified as follows. 
	\begin{definition}[Deviation Gain]\label{Def:deviation_gain}
		For any stage game, denote $i$'s maximum deviation gain at action profile $(a,a)$ by $G(a):=\max\nolimits_{a_i} \pi_i(a_i,a)-\pi_i(a,a)$.
	\end{definition}
	After $i$ deviates, both $i$ and $j$ will switch to the defection action $a^N$, thus $i$ may  suffer from a retaliation loss which is quantified in Definition \ref{Def:retaliation_loss}. 
	\begin{definition}[Retaliation Loss]\label{Def:retaliation_loss}
    For any stage game, denote $i$'s retaliation loss at action profile $(a,a)$ by
    $
        	L(a):=\pi_i(a,a)-\pi_i(a^N,a^N).
    $
    \end{definition}
 
	 Without loss of generality, we assume $G(a)$ is strictly increasing on $a\in[a^N,a^*]$ and non-decreasing for $a>a^*$ (same assumption made in \cite{kamada2020revision}).	

	If any player deviates from $x$ at $-t$, by Definition \ref{Def:deviation_gain}, her deviation gain is $G\big(x(t)\big)$, which materializes only if there is no revision opportunity in the remaining time length $t$. This happens with probability $e^{-\lambda t}$, so the \emph{\textbf{expected deviation gain}} at time $-t$ is $G(x(t))e^{-\lambda t}$. According to Definition \ref{Definition:LR_strategy}, the deviation is followed by a mutual retaliation with duration $kt$.
	During this period, at any time $-s$, if there is any revision opportunity, both players will use $a^N$ instead of $x(s)$. 
	This happens with probability density $\lambda e^{-\lambda s}$. Therefore, the
	\emph{\textbf{expected retaliation loss}} is given as an integral {\small{${\int_{-t}^{-t+kt} \big(\pi(x(s))-\pi^N\big)\lambda e^{-\lambda s} \mathrm{d}s}=\int^{t}_{t-kt} L\big(x(s)\big)\lambda e^{-\lambda s} \mathrm{d}s$}}.

	The above expected deviation gain and expected retaliation loss  jointly determines whether the one-shot deviation at $-t$ is profitable. By comparing these two values and applying Lemma \ref{lemma:one-shot deviation principle}, we obtain the following  result.

	



\begin{theorem}[SPE Constraint]\label{IC SPE iff}
An LR strategy $\sigma=\{x(\cdot),k\}$ constitutes SPE, iff at any time $-t\in [-T,0]$,
{{    	\begin{equation}\begin{aligned}\label{LR incentive constraint}
\underbrace{G\big(x(t)\big)e^{-\lambda t}}_{\text{expected deviation gain}} \leq \underbrace{\int^{t}_{t-kt} L\big(x(s)\big)\lambda e^{-\lambda s} \mathrm{d}s}_{\text{expected retaliation loss}}.
\end{aligned}\end{equation}}}
\end{theorem}

The foundation of the proof is one-shot deviation principle tailored to revision games. Based on this, we analyzed the payoff increase (i.e., expected deviation gain) and the payoff decrease (i.e., expected retaliation loss) caused by one-shot deviation at every time. The equilibrium condition in Theorem \ref{IC SPE iff} is finally obtained by comparing these two kinds of payoff changes and applying one-shot deviation principle.  Theorem 1 says, for an LR strategy to constitute SPE, the expected deviation gain for the one-shot deviation at $-t$ should be no larger than the expected retaliation loss for it.

We emphasize that, for an LR strategy, if Eq. (\ref{LR incentive constraint}) holds for \textit{every} time $-t$, then there exists no profitable one-shot deviation, or equivalently, this LR strategy admits SPE. Specifically, for any fixed $-t$, Eq. (\ref{LR incentive constraint}) does not need to contain any time after $-t+kt$. This is  because for LR,
the one-shot deviation at $-t$ only accounts for payoff changes from $-t$ to $-t+kt$. Any payoff changes after $-t+kt$ are caused by other deviations, which are omitted according to OSDP.


A symmetric profile of LR strategies satisfying these constraints is a \textit{robust SPE, in the sense that even when there is an intentional or unintentional deviation, players can revert to mutual cooperation}.
For a verification of Theorem \ref{IC SPE iff}, please refer to the Appendix.



\section{Equilibrium Computation}\label{sec LR_k plan}
Based on Theorem \ref{IC SPE iff}, to make LR strategies satisfy the SPE constraint, it is equivalent to satisfy Eq.(\ref{LR incentive constraint}) for every time point $-t$. In this section, we will provide a computational method to derive the equilibrium LR strategies governed by Theorem \ref{IC SPE iff}. This method will characterize the cooperative plan for equilibrium LR strategies
in an explicit form.

	
\subsection{Monotone and Piecewise Constant Plan} \label{sec optimality}
We concentrate on LR strategies with cooperative plan $x$ that is in the form of monotone and piecewise constant (MPC) functions. 
For ease of notation, define $\kappa=1-k$.
To derive an MPC plan, divide time span $(-T,0]$ into many small slots: {\small{$(-T,-\kappa T]$, $(-\kappa T, -\kappa^2 T], \cdots, (-{{\kappa}^{n-1}}T, -{{\kappa}^n}T], \cdots$}}
For the $n$-th slot, assign the cooperative plan $x$ with a constant action $a_n$. 
 
{\small{	\begin{equation}\label{Eq:plan}
    	x(t)=\left\{\begin{array}{ll}
    	a_1,& -t\in(-T,-\kappa T]\\
    	a_2,& -t\in(-\kappa T,-\kappa^2 T]\\
    	\vdots &\\
    	a_n,& -t\in(-{{\kappa}^{n-1}}T, -{{\kappa}^n}T]\\
    	a_{n+1},& -t\in(-{{\kappa}^n}T, -{{\kappa}^{n+1}}T]\\    	
    	\vdots &\\
    	\end{array}\right.
	\end{equation}}}We focus on monotonically decreasing MPC plans.
	$a_n$ starts at time {\small{$-{{\kappa}^{n-1}}T$}} and ends at time {\small{$-{{\kappa}^n}T$}}, lasting for a duration {\small{$k\kappa^{n-1}T$}}. This duration shrinks as $n$ increases, meaning that \textit{actions change more frequently as the deadline approaches}. Figure \ref{fig:constant_func} shows an illustration of MPC plan.


	The next step is to satisfy Eq.(\ref{LR incentive constraint}) for each $a_n$ by induction. According to Eq.(\ref{Eq:plan}), any deviation from $a_n$ must be at a time {\small{$-t\in(-{{\kappa}^{n-1}}T, -{{\kappa}^n}T]$}}. Then a mutual retaliation lasts for a duration $kt$. The end of retaliation is $-t+kt=-\kappa t$, which surpasses the end of $a_n$, i.e., {\small{$-\kappa t \ge -{{\kappa}^n}T$}}. Therefore, the constraint for $a_n$ is decomposed as follows:  
{\small{   \begin{equation}
    \begin{aligned}\label{eq.a_i_incentive_constraint}
        G(a_{n})e^{-\lambda t}  \leq\int^{t}_{{{\kappa}^{n}}T} L(a_{n})\lambda e^{-\lambda s} \mathrm{d}s
        +\int^{{{\kappa}^{n}}T}_{\kappa t} L(a_{n+1})\lambda e^{-\lambda s} \mathrm{d}s.
    \end{aligned}  \end{equation}}}The first integral calculates the expected retaliation loss for $[-t, {-{{\kappa}^{n}}T}]$, when the planned action is still $a_n$.  The second integral is the expected retaliation loss for $[{-{{\kappa}^{n}}T}, -\kappa t]$, when the planned action already becomes $a_{n+1}$.

	\begin{figure}[h]
	    \centering
    	    \includegraphics[scale=0.45]{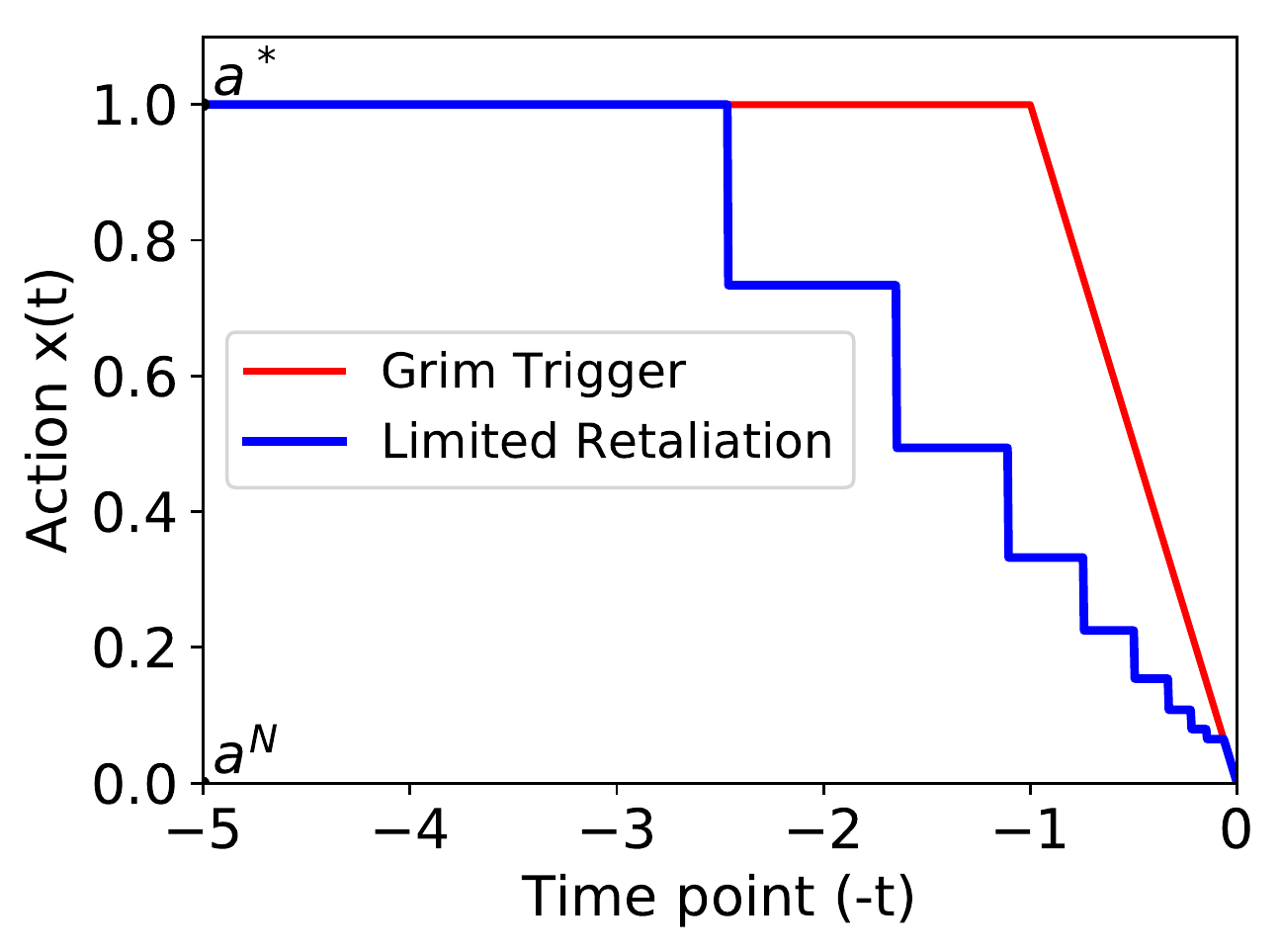}
	    \caption{Example of MPC cooperative plan $x(t)$. }
	    \label{fig:constant_func}
	\end{figure}
\subsection{Inductive SPE Constraint Satisfaction}
For an LR strategy to constitute SPE, the 
Constraint Eq.(\ref{eq.a_i_incentive_constraint}) should hold for every $a_n,$ where $n\ge 1$.
We do this constraint satisfaction by backward induction. The two components of induction are recurrence relation and terminal condition.

\subsubsection{Recurrence Relation}
For an arbitrary $a_n$, Eq.(\ref{eq.a_i_incentive_constraint})  should hold for bounds $-t\to-{{\kappa}^{n-1}}T$ and $-t=-{{\kappa}^n}T$. 
Replacing $-t$ with $-{{\kappa}^n}T$ gives
\begin{equation}\label{upperboundIC}
   G(a_{n})e^{-\lambda {{\kappa}^n}T}  \leq\int^{{{\kappa}^{n}}T}_{ {{\kappa}^{n+1}}T} L(a_{n+1})\lambda e^{-\lambda s} \mathrm{d}s.
\end{equation}
%

Denote $\tau(n+1)=k\kappa^n T$ as the length of the $(n+1)$-th time slot. Expanding the integral in the inequality, we have{\small{\[G\left( {{a_n}} \right){e^{ - \lambda {\kappa ^n}T}} \le L\left( {{a_{n + 1}}} \right)\left( {{e^{ - \lambda {\kappa ^{n + 1}}T}} - {e^{ - \lambda {\kappa ^n}T}}} \right).\]}}
Dividing both sides by ${e^{ - \lambda {\kappa ^{n + 1}}T}}$ gives{\small{$$G\left( {{a_n}} \right)\frac{{{e^{ - \lambda {\kappa ^n}T}}}}{{{e^{ - \lambda {\kappa ^{n + 1}}T}}}} \le L\left( {{a_{n + 1}}} \right)\left( {1 - \frac{{{e^{ - \lambda {\kappa ^n}T}}}}{{{e^{ - \lambda {\kappa ^{n + 1}}T}}}}} \right), $$}}
which is equivalent to 
$$G\left( {{a_n}} \right){e^{(\kappa  - 1)\lambda {\kappa ^n}T}} \le L\left( {{a_{n + 1}}} \right)\left( {1 - {e^{(\kappa  - 1)\lambda {\kappa ^n}T}}} \right).$$
Replacing $\kappa- 1$ by $k$ and then replacing $k{\kappa ^n}T$ by $\tau \left( {n + 1} \right)$, we have the \textit{recurrence relation} as follows.

\begin{equation}\label{Eq:recurrence_relation}
G(a_n)e^{-\lambda\tau(n+1)}
\leq L(a_{n+1})\big(1-e^{-\lambda\tau(n+1)} \big).
\end{equation}
Recall that $
e^{-\lambda\tau(n+1)}
$ is the probability no revision opportunity arrives in the $(n+1)$-th slot. Thus the LHS is the expected deviation gain in this slot. 
Similarly, the RHS represents the expected retaliation loss in the $(n+1)$-th slot. 
Based on Eq.(\ref{Eq:recurrence_relation}), we can generate each $a_n$ inductively. 

\subsubsection{Terminal Condition}
For the terminal action, consider when time approaches the deadline, i.e., $n = c$ where $c$ is large enough. We call $(-\kappa^c T, 0]$ the ultimate slot, and $(-\kappa^{c-1} T, -\kappa^c T]$ the penultimate slot. For the penultimate slot, the planned action is $a_c$, which should satisfy Eq.(\ref{eq.a_i_incentive_constraint}). Substituting the left time bound $-\kappa^{c-1} T$ into it, we have 
\begin{equation}
    \begin{aligned}\label{eq.a_i_incentive_constraint_lowerbound}
        G(a_{c})e^{-\lambda t}  \leq\int^{\kappa^{c-1} T}_{{{\kappa}^{c}}T} L(a_{c})\lambda e^{-\lambda s} \mathrm{d}s.
    \end{aligned}  
\end{equation}
Let $\tau(c)=k\kappa^{c-1} T$ denote the length of the penultimate slot. Basic calculation of this integral:
\begin{equation}\label{Eq:constraint_penultimate_action}
\begin{aligned}
G(a_c)e^{-\lambda\tau(c)}
\leq L(a_{c})\big(1-e^{-\lambda\tau(c)} \big).
\end{aligned}
\end{equation}
This inequality is similar to Eq.(\ref{Eq:recurrence_relation}). But the difference here is that on both sides, the variables are $a_c$. Then the feasible range of $a_c$ can be obtained by solving this inequality.

\subsection{Computing Cooperative Plan for $t \rightarrow 0$}
\subsubsection{Limit Problem at the Deadline}
In the main text, to compute a cooperative action for LR strategy when time is close to the deadline, we define the ultimate slot and penultimate slot, and then set the penultimate slot as that for the terminal action. 
The reason for defining such two adjacent final slots is that:
when time is extremely close to the deadline $0$, in Eq.9, the upper limit and the lower limit of the integral become so close, which makes Eq.9 have no solution.  It is worth noting that the limit problem is a common challenge in many disciplines, which often makes the mathematical tools invalid and makes things  ugly.

To solve this problem in revision games, we (\textbf{I}) find the last $–t$ which is close enough to $0$ but still makes Eq.9 solvable. We call this slot the \textit{penultimate slot}, implying that this is not the hard \textit{ultimate} slot. This \textit{penultimate slot} is the actual terminal condition for our backward induction. (\textbf{II}) However, to make the definition of plan $x$ complete, we introduce the \textit{ultimate} slot back in Eq.11. Here we relax the equilibrium constraint only for the \textit{ultimate} slot, and make it approximately solvable.

\subsubsection{Alternative Solution with Continuous Function}
Besides the approximate equilibrium concept to calculate the plan $x(t)$ for this very short slot, here we discuss about another approach for the incentive constraint satisfaction in the time limit. If $c$ is large enough, this period is very close to the limit and its length is a very small quantity where the piecewise constant plan degenerates to a time-continuous differential plan. Thus we can borrow the grim trigger plan for this very short slot. According to [Kamada and Kandori 2020b], the gradient of this plan $x_g$ is given as 
\begin{equation}\label{eq:GT equation}
\frac{d x_g}{d t}=\frac{\lambda\cdot\left(G(x_g)-L(x_g)\right)}{D'(x_g)}, \quad \forall -t\in (-\kappa^c T, 0].
\end{equation}
Therefore, as long as the grim trigger strategy contains an action $a_c=x_g(\kappa^{c-1} T)$ satisfying the constraint in Eq.(8), we can use this $a_c$ as the terminal action and construct a piecewise constant plan. Thereby, the LR strategy can be obtained.

\subsection{Limited Retaliation in Subgame Perfect Equilibrium}
Considering the recurrence relation together with the terminal action, we directly have the following theorem. 
\begin{theorem}\label{th1}
A profile of Limited Retaliation strategies $\sigma=\{x(\cdot),k\}$ is a subgame perfect equilibrium (SPE), if \\(1) Plan $x$ has the piecewise constant form as Eq.(\ref{Eq:plan}); \\(2) The actions given by $x$ follow recurrence relation Eq.(\ref{Eq:recurrence_relation}); \\(3) There exists an action given by $x$ satisfying Eq.(\ref{Eq:constraint_penultimate_action}).
\end{theorem}
This theorem is just from the principle of mathematical induction  that when the recurrence relation and the initial term (terminal condition, in our case) are given, every member in the sequence can be generated inductively. 

When time is extremely close to the deadline, the terminal condition Eq.(\ref{Eq:constraint_penultimate_action}) becomes more difficult to satisfy, and it degenerates to a limit problem. To see this, divide both sides by $e^{-\lambda\tau(c)}$, then the RHS is decreasing in $c$. To tackle this limit problem, we can use the \emph{approximate} equilibrium solution. The terminal condition is relaxed to
\begin{equation}\label{lowerboundIC_no_integral_2_approximiate}
\begin{aligned}
G(a_c)e^{-\lambda\tau(c)}
\leq L(a_{c})\big(1-e^{-\lambda\tau(c)} \big)+\epsilon,
\end{aligned}
\end{equation}
where $\epsilon$ is a small constant, but $c$ is not necessarily extremely large. Eq.(\ref{lowerboundIC_no_integral_2_approximiate}) is easier to be satisfied. 

\begin{corollary}\label{corollary1}
A profile of LR strategies strictly satisfying conditions (1) and (2) but approximately satisfying condition (3) in Theorem \ref{th1} is an approximate SPE.
\end{corollary}

\begin{definition}
With a given $k$, if a plan $x$ can make a profile of LR strategies constitute SPE, we say plan $x$ supports SPE. 
\end{definition}

Note that a plan such that $\forall t, x(t) = a^N$ trivially satisfies Eq.(\ref{Eq:recurrence_relation}) and Eq.(\ref{Eq:constraint_penultimate_action}). Thus an LR strategy using this trivial plan can support SPE. However, this one cannot sustain any level of cooperation. For LR strategies to constitute a cooperative SPE, we have the following requirements. 
\begin{corollary}\label{corollary2}
A cooperative plan $x$ in the form of Eq.(\ref{Eq:plan}) supports SPE, only if for all $n$, $a_n\geq a^N$. 
\end{corollary} 
\begin{proof}
Eq.(\ref{Eq:recurrence_relation}) can identify the equilibrium constraint for the $(n+1)$-th slot. Divide both sides by $e^{-\lambda \tau(n+1)}$, we have 
$G(a_n)
\leq L(a_{n+1})\big(e^{\lambda\tau(n+1)}-1 \big).$
$L(a)$ is monotonically increasing in $a < a^*$, and if $a < a^N$, $L(a) < L(a^N) = 0$. $G(a) \geq 0$ for any $a$. To make Eq.(\ref{Eq:recurrence_relation}) hold, it is required that $a_n \geq a^N$ for any $n \geq 1$.    
\end{proof}

\begin{corollary}\label{corollary3}
A cooperative plan $x$ in form of Eq.(\ref{Eq:plan}) supports cooperative SPE, only if the terminal action $a_c> a^N$.
\end{corollary}
\begin{proof}
Prove by contrapositive. Divide both sides of Eq.(\ref{Eq:recurrence_relation}) by $e^{-\lambda \tau(n)}$, we have 
$G(a_n)
\leq L(a_{n+1})\big(e^{\lambda\tau(n)}-1 \big).$
Let $a_c = a^N$, then $L(a_c)=0$. Then the inequality holds only when $G(a_{c-1})=0$. By the same recurrence relation, this further requires $a_{c-1}=a^N$. Repeat this for all $a_n$, we can see that if $a_c = a^N$, then the only plan satisfying Eq.(\ref{Eq:recurrence_relation}) is that $a_n=a^N$ for all $n$, which is not cooperative. The SPE induced by this trivial plan is also not cooperative.  
\end{proof}

In summary, Theorem \ref{th1} presents the condition for LR strategies to constitute SPE, while Corollaries \ref{corollary2} and \ref{corollary3} show requirements for a cooperative plan to be a good choice for LR strategies and for cooperative LR strategies, respectively. 


\section{Social Welfare Maximization}\label{sec3}
In this section, we show that among all piecewise constant plans, the bounded MPC plans maximize the social welfare. We then devise an optimization method to find them. 


\subsection{Reducing Search Space of Plan Optimization} \label{sec optimality}
We will first prove that MPC plans have the highest social welfare among all piecewise constant plans. The following analysis allows us to only focus on \emph{{bounded plans}} which have upper-bound action $a^*$ and lower-bound action $a^N$. 

\subsubsection{Bounding}
According Corollary \ref{corollary2}, we know that
if a plan $x$ supports SPE, it must satisfy that $\forall t, x(t) \geq a^N$.
Now based on $x$, we construct an outer function $\bar{x}$ as follows.
{\small\begin{equation}\label{Eq:xbar}
\forall t, \bar{x}(t) = \left\{\begin{array}{ll}
a^*, & x(t) > a^*, \\
x(t), & otherwise. \\
\end{array}\right.
\end{equation}}$\bar{x}$ is still piecewise constant. Since it replaces all actions larger than $a^*$ by $a^*$, it must be that $\forall t, x(t)\in [a^N,a^*]$, meaning $\bar x$ is bounded. The bounded plans have the following property which allows us to only concentrate on them.

\begin{lemma} \label{Lemma:bounding}
For every piecewise constant plan $x$ that supports SPE, there exists a bounded piecewise constant plan $\bar x$ which also supports SPE, and has higher social welfare.
\end{lemma}
\begin{proof}
For notation simplicity, rewrite the expected payoff in Eq.(\ref{Eq:expected_payoff}) by $V(x)$.
Given a piecewise constant plan $x$, we construct $\bar x$ following Eq.(\ref{Eq:xbar}). First prove $V \big( \bar{x} \big) \ge V \big(x \big)$. Recall $a^*$ results in the maximum stage game payoff. Since $\bar{x}$ always chooses the same or a better action at any time, according to the monotonicity of $\pi(a)$, it must be that $\forall t, \pi(\bar{x}(t)) \geq \pi(x(t))$. The expected payoff $V$ is adding up $\pi(t)$ over all $t$. 
Therefore, $V(\bar{x}) \ge V(x)$. The social welfare is the sum of all players' expected payoffs, thus $\bar x$ leads to a higher social welfare. Then, we prove that $\bar{x}$ supports SPE. 
On the one hand, by Eq.(\ref{Eq:xbar}), at any time $t$, the two actions $\ \bar{x}(t) \leq x(t)$. Since $G(a)$ is strictly increasing on $[a^N,a^*]$, we have $G(\bar{x}(t))\le G(x(t)) $ for any $t$. On the other hand, $L(a)$ is maximized by $a^*$, then by Eq.(\ref{Eq:xbar}) we know $L(x(t))\leq L(\bar{x}(t))$ for any $t$.
The premise indicates that $x(t)$ satisfies equilibrium constraint Eq.(\ref{LR incentive constraint}), where the LHS is decreasing as $G$ decreases and the RHS is increasing as $L$ increases. Replacing $x$ by $\bar x$ still satisfies Eq.(\ref{LR incentive constraint}), which means $\bar x$ also supports SPE.
\end{proof}


\subsubsection{Monotonization}  
The second step will further refine the above result and show that, among all the bounded piecewise constant plans, 
we can focus on the bounded monotone and piecewise constant (bounded MPC) plans. 

    
Consider a bounded piecewise constant plan $\bar x$ which supports SPE but is not monotonically decreasing over $(-T,0]$. Assume two numbers $i<j$. Assume in the action sequence of plan $\bar x$, $a_j $ is the first action which is larger than some actions before it. This means $x$ is not monotonically decreasing. Specifically, assume all actions after $a_i$ and before $a_j$ are smaller than $a_j$. According to Eq.(\ref{Eq:plan}), the planned slots for $a_{i}$ and $a_{j}$ are $(-\kappa^{i-1}T,-\kappa^{i}T]$ and $(-\kappa^{j-1}T,-\kappa^{j}T]$, respectively. Thus the time-span after $a_{i}$ and before $a_{j}$ is $(-\kappa^{i}T,-\kappa^{j-1}T]$.    
Based on the above introduction of the non-monotone $x$, we construct another plan $\hat{x}$ as follows.
{\small
\begin{equation}\label{Eq:hat_x}
\hat{x}(t)=\left
\{\begin{array}{l}
a_j, ~~ -t\in (-\kappa^{i}T,-\kappa^{j-1}T] \\
\bar x(t), ~~otherwise
\end{array}
\right.
\end{equation}}All actions in $\hat{x}$ are from $\bar x$ which is bounded, then $\hat{x}$ is also bounded. Moreover, $\hat{x}$ replaces all the `non-monotone' actions $\hat x(t)$ for $t\in(-\kappa^{i}T,-\kappa^{j-1}T]$ with a constant $a_j$, hence it restores the monotonicity. Then Lemma \ref{Lemma:monotonicity} shows it suffices to only consider the bounded MPC plans.


    
\begin{lemma}
\label{Lemma:monotonicity}
For every bounded piecewise constant plan $\bar x$ which supports SPE, there always exits a bounded monotone piecewise constant (bounded MPC) plan $\hat x$, which also supports SPE and leads to higher social welfare than $\bar x$.
\end{lemma}

\begin{proof}
For ease of notation, We rewrite the expected payoff function $V(\sigma, T)$ into $V(x)$. This will not change the computation. We say $V(\tilde{x}) \succeq V(x)$ if $\forall T \geq 0, V(\tilde{x}, T) \geq V(x, T)$, which means plan $\tilde{x}$ gains no less expected payoff than $x$ on any scale of time period. For short, we say $V(x)$ is the expected payoff of plan $x$.

Given a bounded piecewise constant plan $\bar x$, construct $\hat x$ following Eq.(11). This construction already shows $\hat x$ is monotone. Now we prove that $\hat x$ supports SPE and has higher social welfare than $\bar x$. We prove Lemma 2 by comparing the expected deviation gain and the expected retaliation loss (punishment loss) in different time intervals.
        
Case 1: $-t \in (-\kappa^{j-1}T,0]$. According to Eq.(11), $\hat{x}(t) = \bar x(t)$ in this period. Since $\bar x$ leads to SPE, it is straightforward that $\hat{x}$ satisfies the equilibrium constraint.

Case 2: $-t\in(-\kappa^i T, -\kappa^{j-1}T]$. According to Eq.(11), $\hat{x}(t)= a_j $. A deviation from this action at any time $-t$ in this slot acquires an expected deviation gain $G(a_j) e^{-\lambda t}$. After this deviation, a retaliation is taken until $-t+kt$, and $-t+kt\leq-\kappa^{j}T$, meaning that the end of the retaliation can not surpass that of $a_j$. Therefore, the retaliation is always against $a_j$ and the expected retaliation loss is 
\begin{equation} \label{Eq:dotx_punishment loss}
\begin{aligned}\int_{(1-k)t}^{t} L(a_j) \lambda e^{-\lambda s} \mathrm{d}s=L(a_j) \big(e^{\lambda k t} - 1\big) e^{-\lambda t}.
\end{aligned}
\end{equation}
On the other hand, consider plan $\bar x$ in slot $(-\kappa^{j-1}T, -\kappa^{j}T]$. Since $\bar x$ leads to SPE, it should satisfy Eq.(4) for this slot. Substituting the upper-bound $-\kappa^j T$ into it and multiplying both sides by a factor $e^{-\lambda t}$, we get
\begin{equation} 
\begin{aligned}
G(a_j)e^{-\lambda t}
\leq L(a_{j})\big(e^{\lambda\kappa^j T}-1 \big)e^{-\lambda t}.
\end{aligned}
\end{equation}
Since we are considering slot $(-\kappa^i T, -\kappa^{j-1}T]$, it must be that any $t$ in this slot is larger than $\kappa^j$. 
Therefore, for plan $\hat x$ in slot $(-\kappa^i T, -\kappa^{j-1}T]$, the expected deviation gain $G(a_j)e^{-\lambda t}$ is less than the retaliation loss in Eq.(\ref{Eq:dotx_punishment loss}). The incentive constraint is satisfied in this slot.

Case 3:$-t\in(-\kappa^{i-1}T,-\kappa^{i}T]$. Since $\bar x$ results in SPE, it should satisfy Eq.(4), which is
\begin{equation}\label{Eq:case2constraint}
    \begin{aligned}
        G(a_{i})e^{-\lambda t}  &\leq\int^{t}_{{{\kappa}^{i}}T} L(a_{i})\lambda e^{-\lambda s} \mathrm{d}s
        +\int^{{{\kappa}^{i}}T}_{\kappa t} L(a_{i+1})\lambda e^{-\lambda s} \mathrm{d}s. 
\end{aligned}  
\end{equation}
By Eq.(11), we know $\hat{x}(t) = a_i = \bar x(t) \geq a_j > a_{i+1}$. Together with this ordering, Eq.(\ref{Eq:case2constraint}) leads to 
\begin{equation}
    \begin{aligned}
        G(a_{i})e^{-\lambda t}  
        &< \int^{t}_{{{\kappa}^{i}}T} L(a_{i})\lambda e^{-\lambda s} \mathrm{d}s
        +\int^{{{\kappa}^{i}}T}_{\kappa t} L(a_{j})\lambda e^{-\lambda s} \mathrm{d}s. \\
        &= \int^{t}_{{{\kappa}^{i}}T} L(a_{i})\lambda e^{-\lambda s} \mathrm{d}s
        +\int^{{{\kappa}^{i}}T}_{\kappa t} L(a_{i+1})\lambda e^{-\lambda s} \mathrm{d}s.
\end{aligned}  \end{equation}
The second line is just the retaliation loss of $\hat x$, thus it satisfies equilibrium constraint on this interval.

        
Case 4: $-t\in (-T,-\kappa^{i-1}T]$. During this period, it is always that $\hat x(t)=\bar x(t)$. Also, at the ending time of retaliation, it is still that $\hat x(t)=\bar x(t)=a_i$. Thus $\hat x(t)$ satisfies the equilibrium constraint as $\bar x(t)$ does.

     
The above four cases cover all possible deviations. In each case, $\hat x$ satisfies the equilibrium constraint, thus $\hat x$ leads to SPE. Moreover, we also know that $ \forall t \geq 0,\hat{x}(t) \geq \bar x(t)$, thus $V\big(\hat{x}\big) \succeq V\big(\bar x\big)$.     
\end{proof}

    
With Lemmas \ref{Lemma:bounding} and \ref{Lemma:monotonicity}, we can obtain Theorem \ref{Theorem:bounded and monotone} which reduces the optimization space as much as possible.
\begin{theorem}\label{Theorem:bounded and monotone}
Among all piecewise constant plans supporting SPE, a bounded MPC plan has maximum social welfare.
\end{theorem}




\subsection{Social Welfare Maximizing MPC Plan}
The previous subsection shows the bounded MPC plans have the highest social welfare among all piecewise constant plans. In this subsection, we will finally find the social welfare maximizing plan within this domain.

	

\subsubsection{KKT Conditions}
From Theorem \ref{th1} and Theorem \ref{Theorem:bounded and monotone} we know that, given a fixed constant number $c$, the cooperative plan $x=\{a_1,\cdots,a_n, \cdots, a_c\}$ can maximize the social welfare $V(x)$ if and only if it satisfies Eq.(\ref{Eq:recurrence_relation}) and Eq.(\ref{eq.a_i_incentive_constraint_lowerbound}), and is bounded and monotone. Thus the social welfare maximizing problem is formally represented as follows. 
\begin{argmaxi}
{\scriptstyle x=\left\{a_1,\cdots,a_c\right\}}{V(x)=\sum_{n=1}^{c}\pi\big(a_n\big)\lambda e^{-\lambda n} }{}{}
\addConstraint {G(a_n)
\leq L(a_{n+1})\big(e^{\lambda\tau(n)} -1\big)} 
\addConstraint{ G(a_c) 
\leq L(a_{c})\big(e^{\lambda\tau(c-1)} -1\big)}
\addConstraint{a_{n+1}\le a_n \le a^*}.
\end{argmaxi}\label{KKT_big}This is a typical non-linear optimizations problem which has Karush-Kuhn-Tucker (KKT) optimality conditions. In general, the optimal solution can be obtained by the Lagrange multiplier method. 
However, the computation cost could be high. Next, we further simplify this maximization into smaller problems.

\begin{figure*}[htp]
    \centering

    \subfigure[Error rate $2\%$, $k=0.33$]
    {
        \includegraphics[scale=0.32]{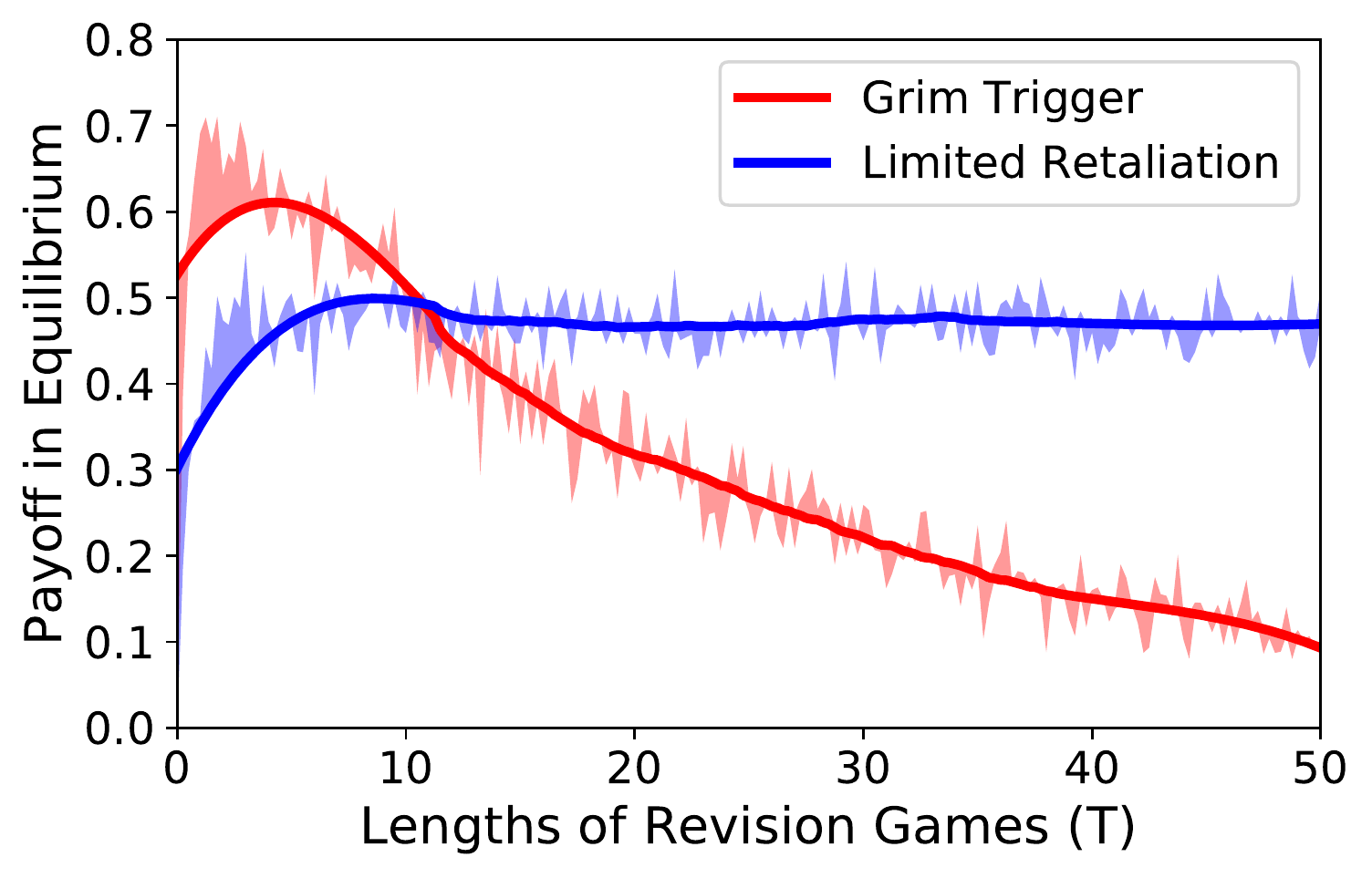}
        \label{fig:third_sub}
    }
    \subfigure[Error rate $10\%$, $k=0.33$]
    {
        \includegraphics[scale=0.32]{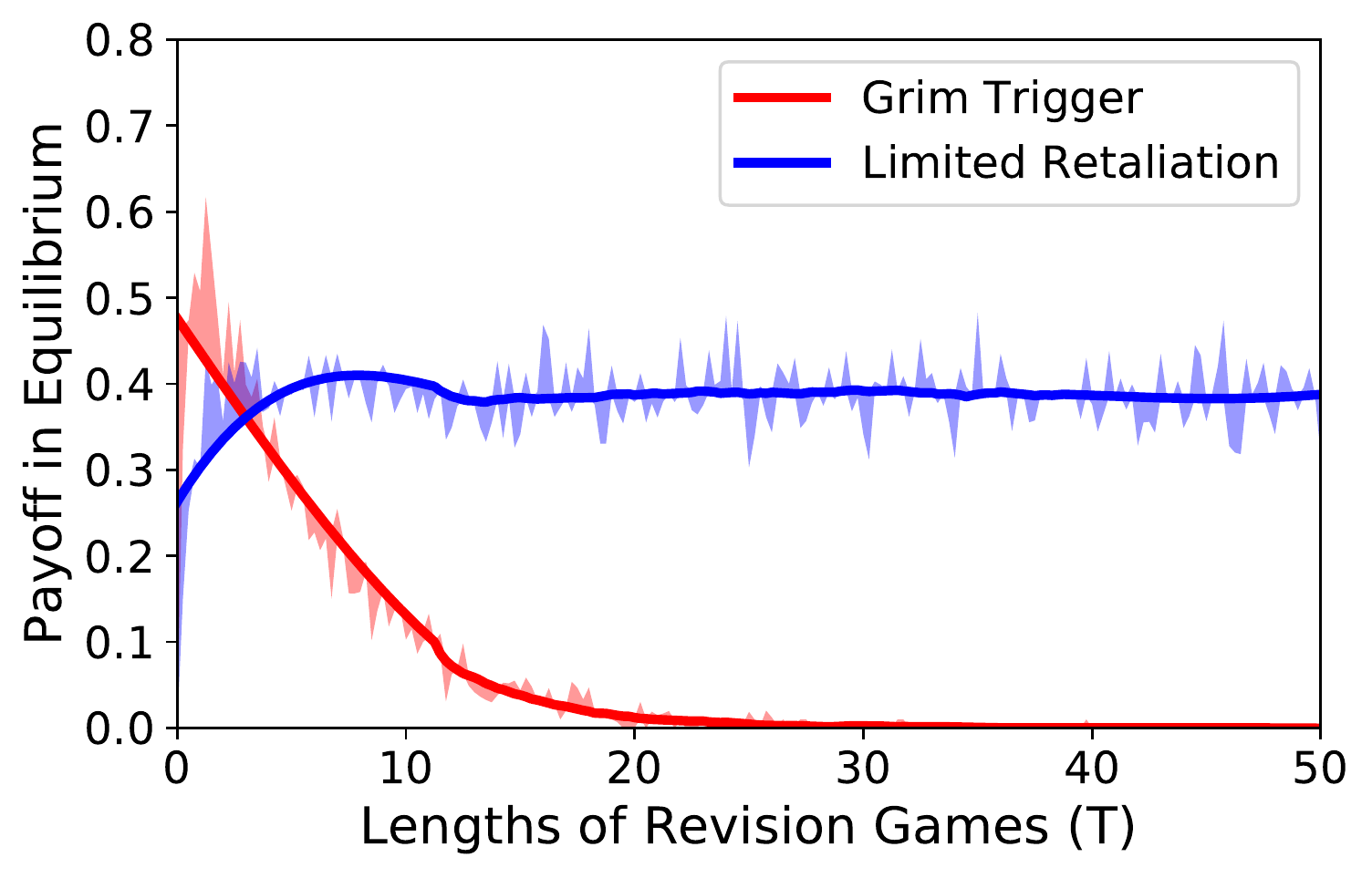}
        \label{fig:third_sub}
    }
    \subfigure[Error rate $30\%$, $k=0.33$]
    {
        \includegraphics[scale=0.32]{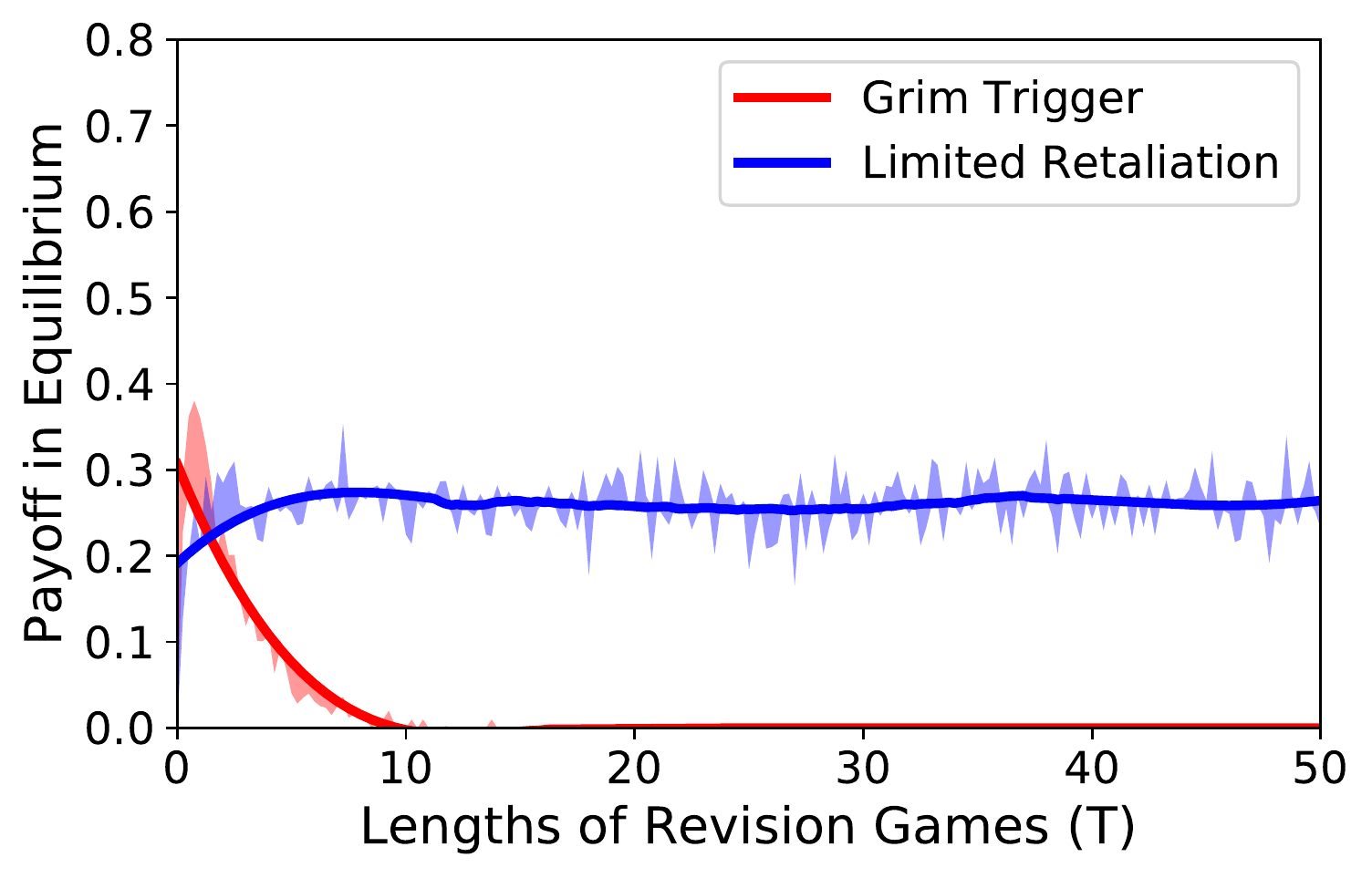}
        \label{fig:third_sub}
    }
    \caption{Equilibrium payoffs of LR and GT in revision Prisoner's Dilemma with different error rates. In each subfigure, $50$ different revision games are simulated, with lengths varying from $T=0^+$ to $T=50$. The $y$-axis is the expected equilibrium payoff. 
	Payoff in each simulated game is depicted by a colored dot, and all dots from $50$ revision games compose a curve. 
	}
    \label{fig:pd}
\end{figure*}
\subsubsection{Inductive Payoff Maximization}
From Theorem \ref{Theorem:bounded and monotone}, the optimization space has been reduced to bounded MPC plans. Eq.(14) can be transformed into many sub-problems.
\begin{theorem}\label{prop:KKTsub}
With a fixed terminal action $a_c$ satisfying Eq.(\ref{Eq:constraint_penultimate_action}),
the KKT optimility conditions Eq.(14) can be decomposed into a series of sub-problems for different $a_n$. Each sub-problem is characterized as follows:
\begin{argmaxi}
{\scriptstyle a_n}{\pi\big(a_n\big)}{}{}
\addConstraint {G(a_n)
\leq L(\dot{a}_{n+1})\big(e^{\lambda\tau(n)} -1\big)} 
\addConstraint{\dot{a}_{n+1}=\arg\max \pi\big(a_{n+1}\big)}
\addConstraint{\dot{a}_{n+1}\le a_n \le a^*}.
\end{argmaxi}
The inductive solution of the series of sub-problems Eq.(15) is identical to that of the original KKT conditions Eq.(14).
\end{theorem}
\begin{proof}
 Let $\tilde{x}= \{\tilde{a}_1, \dots, \tilde{a}_c\}$ be the optimal solution of Eq.(14) and let $\dot{x}=\{\dot{a}_1, \dots, \dot{a}_c\}$ be the solution of a series of sub-problems in Eq.(15). Note that $\tilde{a}_c=\dot{a}_c$, both determined by Eq.(\ref{Eq:constraint_penultimate_action}). We prove $\tilde{x}=\dot{x}$ by contradiction. 
 
Assume $\tilde{x}$ disagrees with $\dot{x}$ at some actions. Since the two terminal actions  $\tilde{a}_c=\dot{a}_c$, there exists at least one time slot, after which $\tilde{x} $  is identical to $\dot{x}$. Let $n = \arg\max \{\tilde{a}_n \ne \dot{a}_n\}$. Since $\tilde{a}_{n+1}= \dot{a}_{n+1}$, then by Eq.(\ref{Eq:recurrence_relation}), $\tilde{a}_n$ and $\dot{a}_n$ have the same constraints. $\pi$ is strictly increasing, by  $\arg\max \pi$ in Eq.(15), $\dot{a}_n > \tilde{a}_n$. Since $L(a)$ is strictly increasing, then $G(\tilde{a}_{n-1}) \leq L(\tilde{a}_n)(e^{\lambda \tau (n)}-1) < L(\dot{a}_n)(e^{\lambda \tau (n)}-1).$
Thus another plan $
    \bar{x} = \{\tilde{a}_1, \dots, \tilde{a}_{n-1}, \dot{a}_n, \tilde{a}_{n+1}, \dots, \tilde{a}_c\}
$ is also a solution for Eq.(14) and $V(\bar{x}) > V(\tilde{x})$ by monotonicity of $\pi(a)$. This contradicts the premise that $\tilde x$ is optimal.
\end{proof}

By applying Lagrange multiplier method inductively, we can solve the series of Eq.(13) and obtain the social welfare maximizing MPC plan. The expected payoff for the optimal LR strategy with this plan is simply $V(\dot{x})=\sum_{n=1}^{c}\pi\left(\dot{a}_n\right).$

\section{Numerical Results}

	The existing methods of revision games are all based on the grim trigger (GT) strategy \cite{kamada2020revision}, where as long as there is any deviation (either by unintentional action error or intentional deviation), both players will never forgive, and retaliate to the deadline. In this section, we examine the performance of LR strategies, and compare them with GT strategy in Prisoner's Dilemma and the famous Cournot game, where players may make occasional mistakes.

\subsection{Limited Retaliation in Prisoner's Dilemma}

	Consider in a continuous Prisoner's Dilemma, player-$i$'s payoff is $\pi_i(a_i,a_{-i})=B(a_{-i})-C(a_{i})$,
	where $B(a)>C(a)$ and $B(a)$ and $C(a)$ denote the player's benefit and cost, respectively \cite{killingback1999variable}. Without loss of generality, we use $B(a)=2a$, $C(a)=a^2$  and $a_i\in A=[0,1]$. Then the payoff of $i$ is: $\pi_i=2a_{-i}-a_i^2$. The Nash action is $a^N=0$ while the fully cooperative action is $a^*=1$. For the revision version of this stage game, we set $\lambda=1 $ for the Poisson arrival rate. For the LR strategy, we set the retaliation fierceness $k=0.33$ as an example. The collusive plan for this LR strategy has been shown in Figure \ref{fig:constant_func}. Now we set different lengths and different error rates of the revision games and show the simulation results in Figure \ref{fig:pd}.
Players' action error rates in the subfigures (a), (b) and (c) are low ($2\%$), medium ($10\%$) and high ($30\%$), respectively.
	
	In each subfigure, each value of $T$ on the $x$-axis represents a total length of a simulated revision game. That is, in each subfigure, $50$ different revision games are simulated, with game lengths varying from $T=0^+$ to $T=50$. The $y$-axis is the expected equilibrium payoff. There are two colored curves in each subfigure, one is for LR strategy equilibrium and the other is for GT strategy equilibrium. The payoff for each strategy in each simulated game is depicted by a colored dot, and all the dots compose a colored curve. For each simulated game (i.e., each dot), we sample players' actions according to the error rate. For example, in subfigure (a), players follow the plan $x$ of LR (or the plan of GT) with a probability $98\%$, but has a probability $2\%$ to make a mistake at \emph{any} revision opportunity. We have two observations from the simulations.
	\begin{enumerate}
	    \item By comparing the two curves in all three subfigures, we can see that,
	    as the error rate increases, LR becomes better and better than GT. Thus in a noisy environment, forgiveness is a crucial component to correct agents' errors and LR can help agents to obtain high payoffs. Evidence of forgiveness has commonly been seen in the real world. Real people are usually not extremely grim. A society in which people are cooperative, vengeful but also forgiving usually have higher social welfare.

	    \item In each subfigure,
	    the payoff of LR is relatively stable for different game lengths (different values on the $x$-axis), while the payoff for GT decreases drastically as the game length increases.
	    Thus for a revision game with a longer duration, LR performs better than GT.
	   Such an observation is coincident with real-world scenarios that in people's long-run interactions, forgiveness plays a more important role for people to go to a win-win outcome. A longer future interaction confers upon people more incentive to forgive a mistake, as long as this mistake can be properly punished.
	\end{enumerate}

\subsection{Impact of Noise Level}
\begin{figure*}
    \centering
    \subfigure[Error rate $2\%$, $k=0.33$]
    {
        \includegraphics[scale=0.3]{pd-k0-33-e0-02.pdf}
        \label{fig:third_sub}
    }
    \subfigure[Error rate $10\%$, $k=0.33$]
    {
        \includegraphics[scale=0.3]{pd-k0-33-e0-10.pdf}
        \label{fig:third_sub}
    }
    \subfigure[Error rate $30\%$, $k=0.33$]
    {
        \includegraphics[scale=0.3]{pd-k0-33-e0-30.pdf}
        \label{fig:third_sub}
    }
    
        \subfigure[Error rate $2\%$, $k=0.60$]
    {
        \includegraphics[scale=0.3]{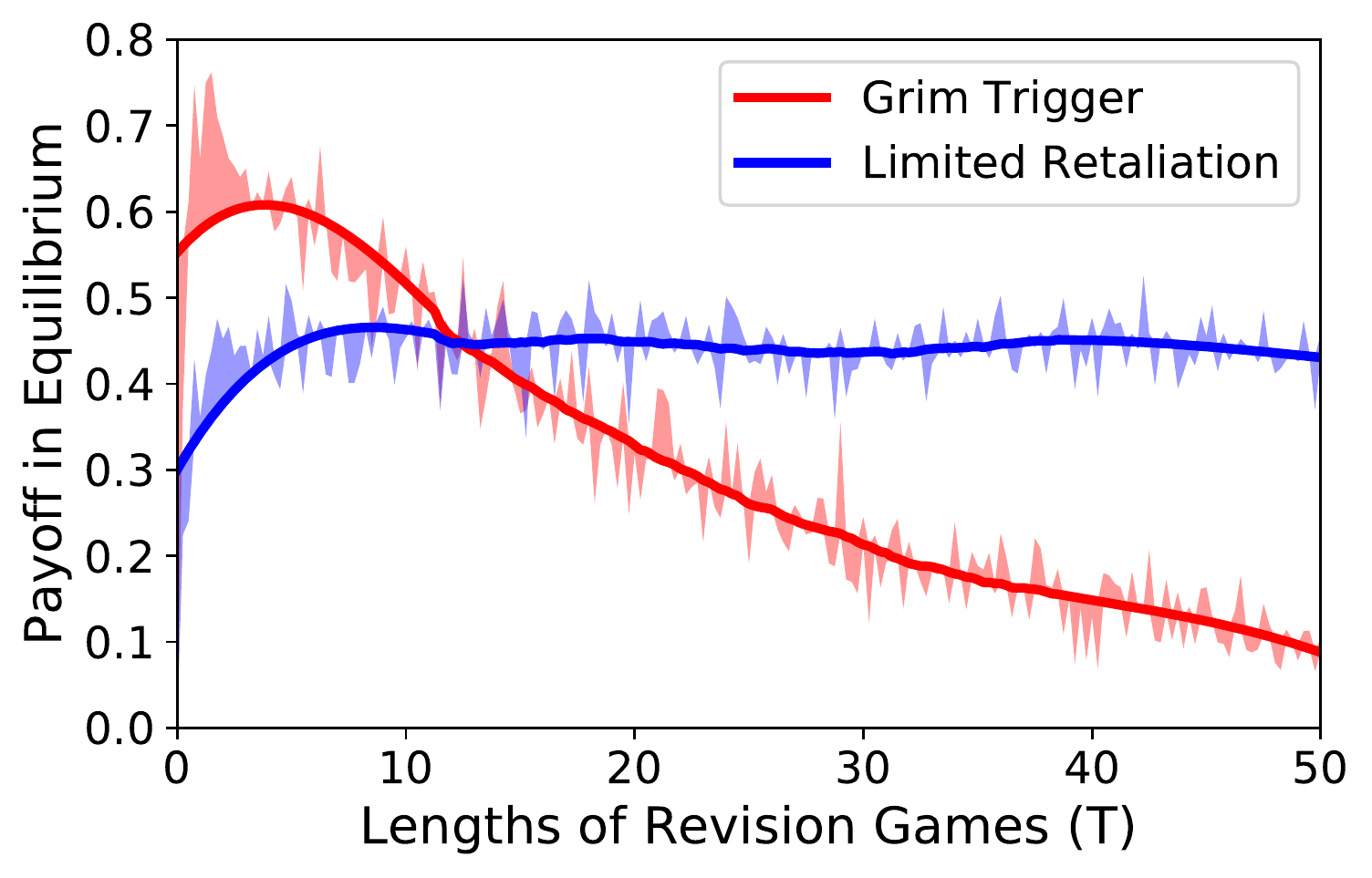}
        \label{fig:third_sub}
    }
    \subfigure[Error rate $10\%$, $k=0.60$]
    {
        \includegraphics[scale=0.3]{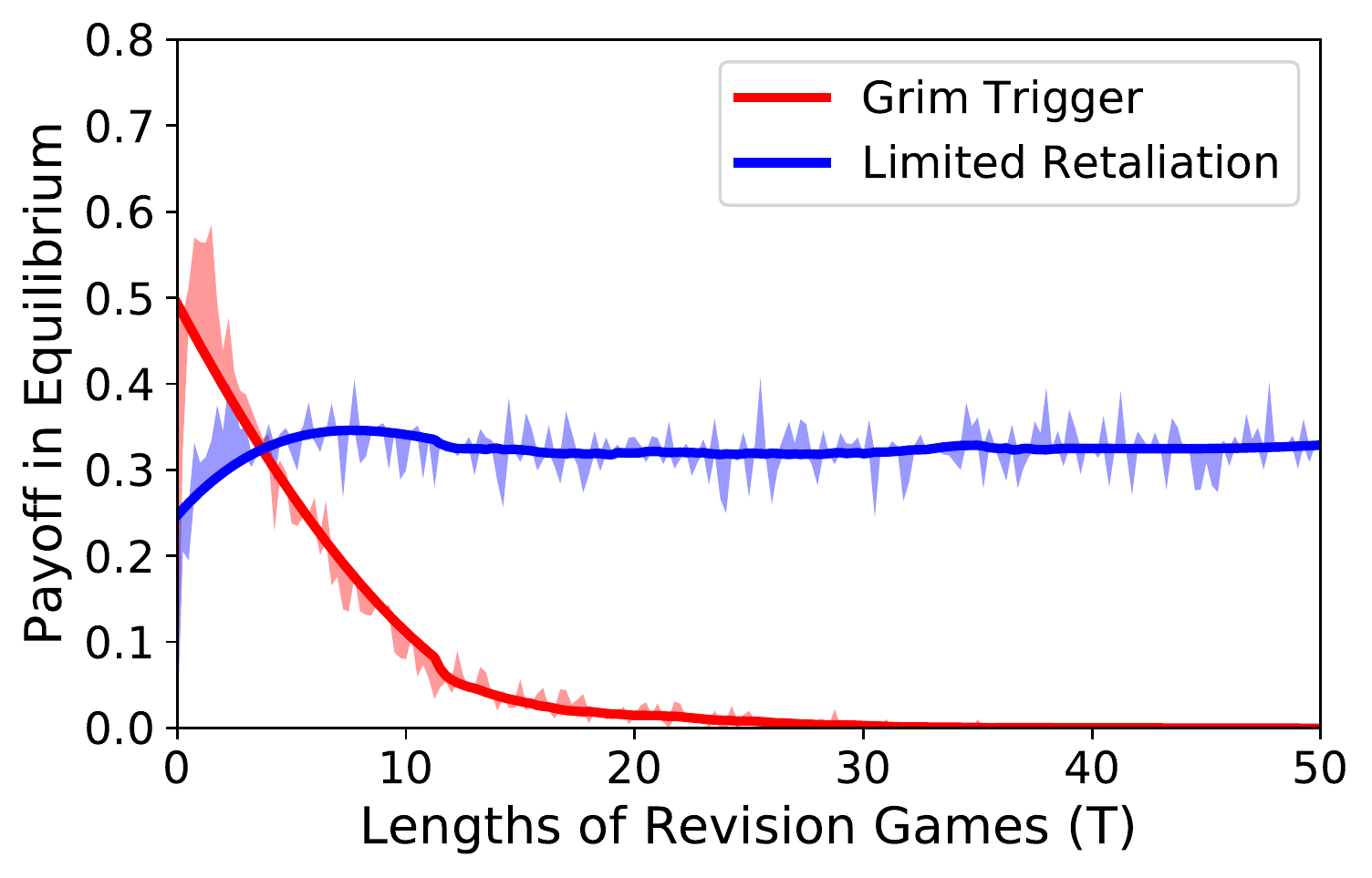}
        \label{fig:third_sub}
    }
    \subfigure[Error rate $30\%$, $k=0.60$]
    {
        \includegraphics[scale=0.3]{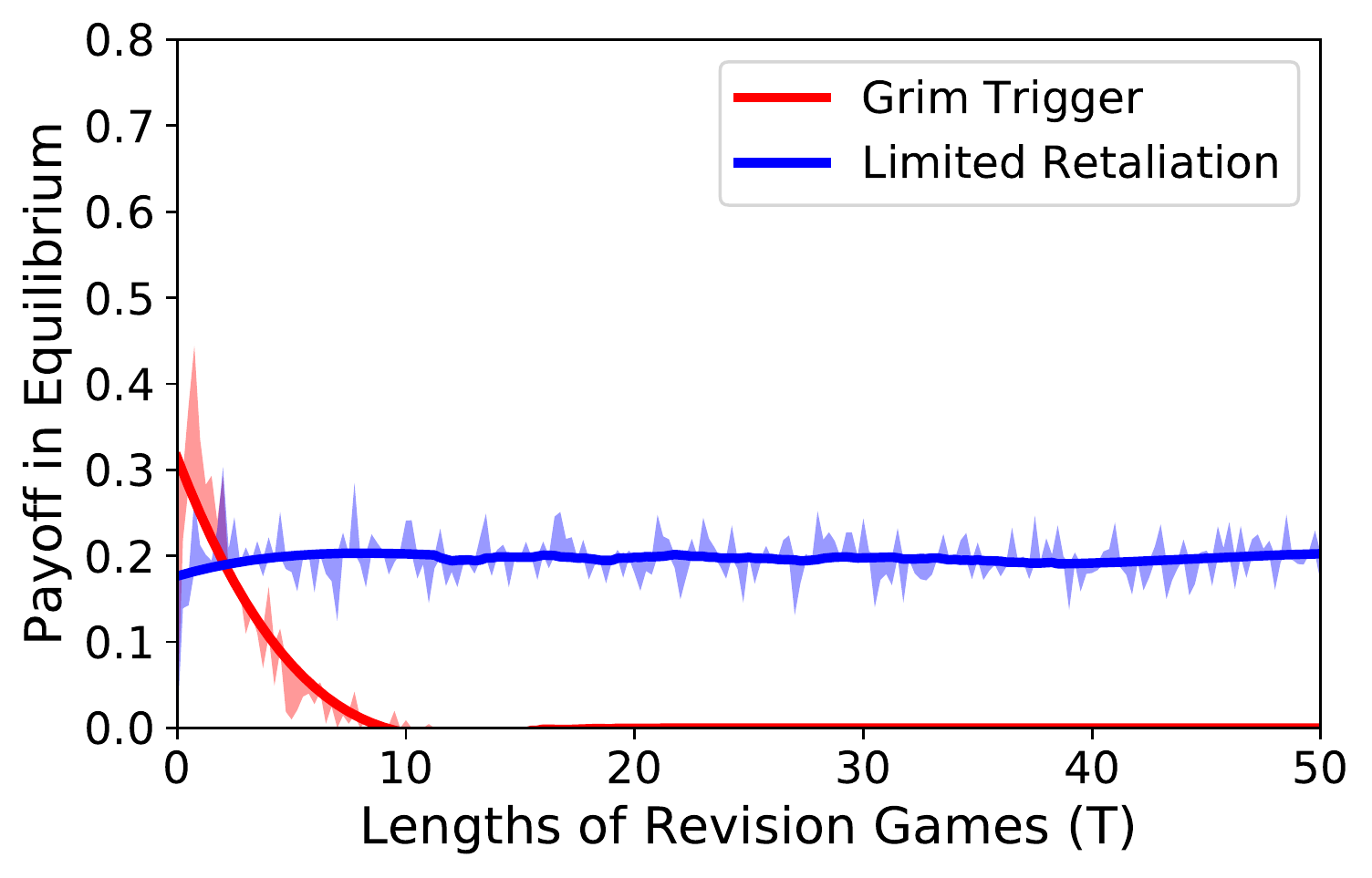}
        \label{fig:third_sub}
    }
    
            \subfigure[Error rate $2\%$, $k=0.85$]
    {
        \includegraphics[scale=0.3]{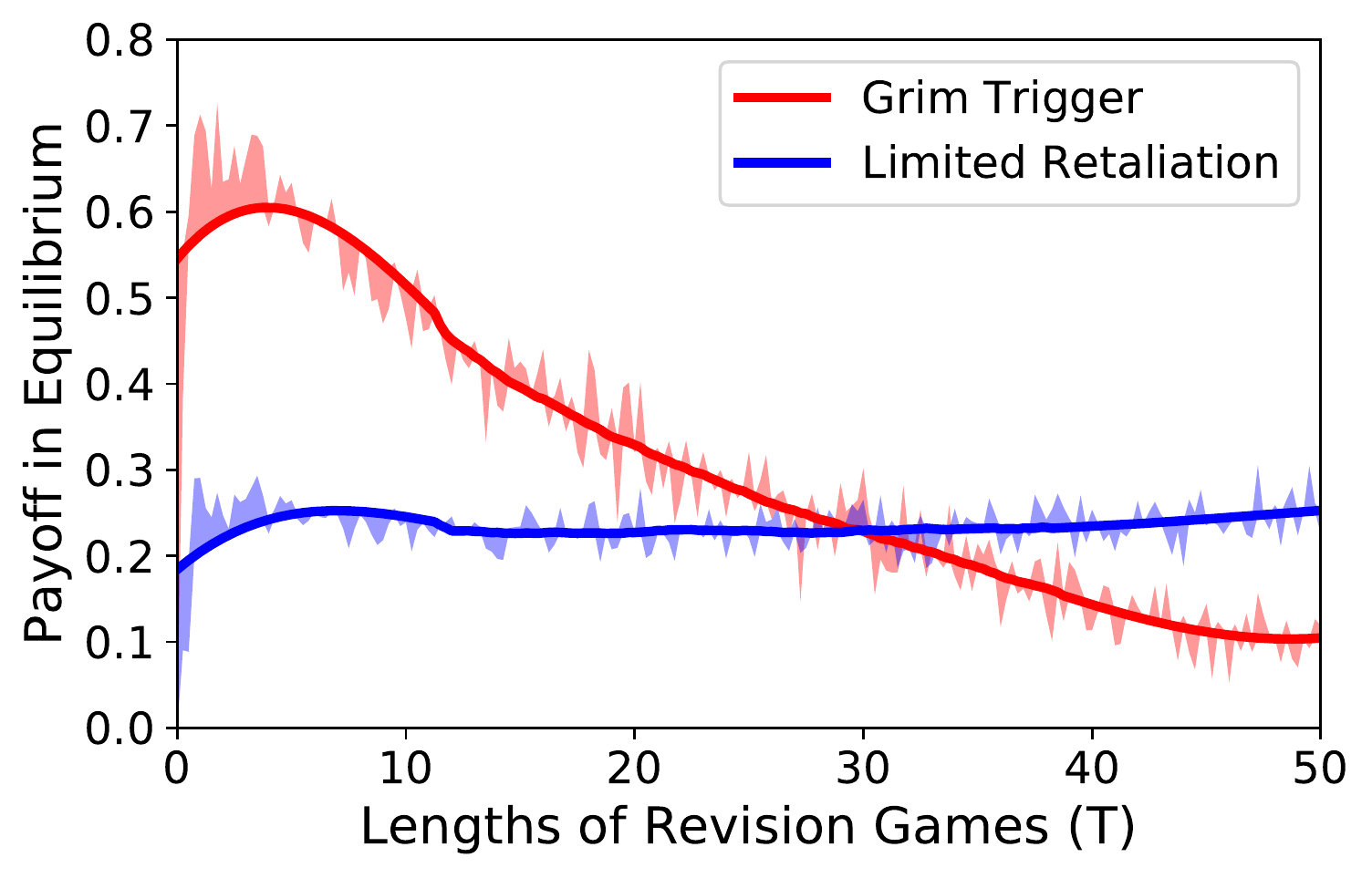}
        \label{fig:third_sub}
    }
    \subfigure[Error rate $10\%$, $k=0.85$]
    {
        \includegraphics[scale=0.3]{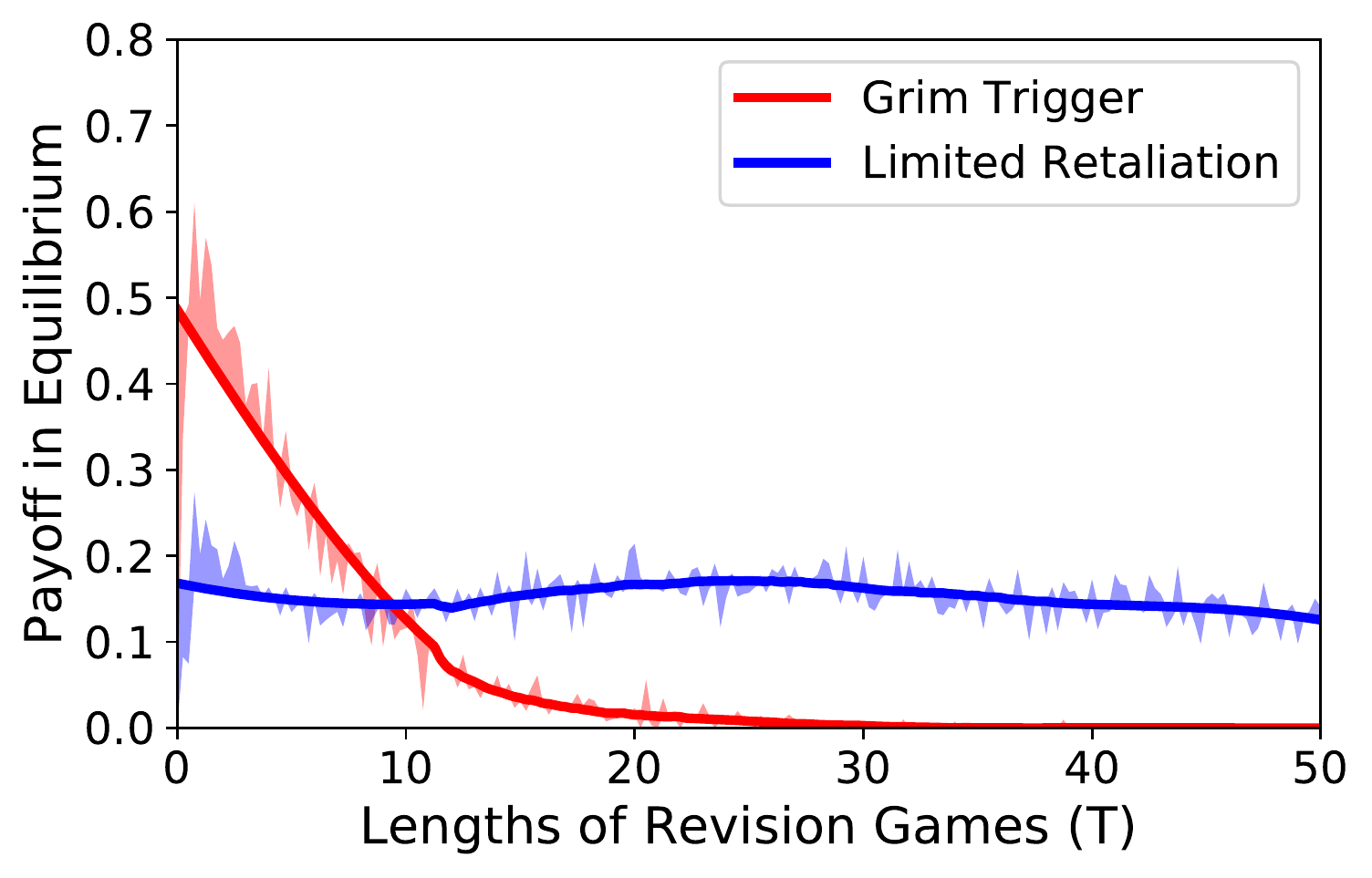}
        \label{fig:third_sub}
    }
    \subfigure[Error rate $30\%$, $k=0.85$]
    {
        \includegraphics[scale=0.3]{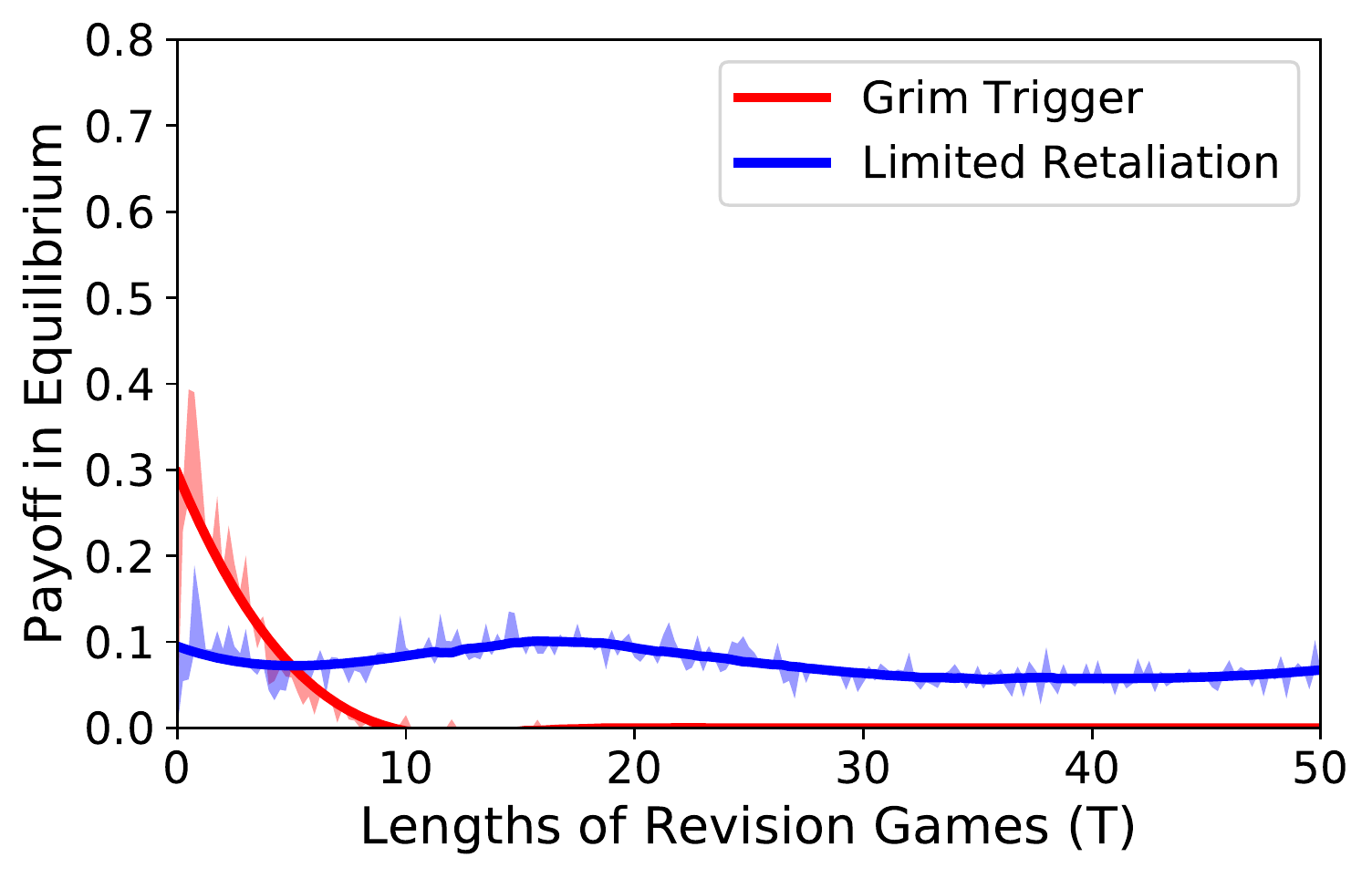}
        \label{fig:third_sub}
    }
    \caption{Equilibrium Payoffs in the Prisoner's Dilemma revision games with different error rates. GT strategy is the same in all subfigures while LR strategies have different retaliation fierceness $k$. }
    \label{fig:pd}
\end{figure*}

In this part, we show more simulations of LR strategies in Prisoner's Dilemma revision games. The results are in Figure \ref{fig:pd}. In the first row, the retaliation fierceness is set as $k=0.33$, which represents a mild retaliation. In the second row, players' retaliations are $k=0.6$, which represents a medium retaliation. In the last row, players are very vengeful, their retaliation fierceness are set to $k=0.85$. The grim trigger strategies is identical in each subfigure.
The error rates in the first, the second and the last columns are $2\%,10\%$ and $30\%$, respectively.

The existing methods of revision games are all based on the grim trigger (GT) strategy, where as long as there is any deviation (either by unintentional action error or intentional deviation), both players will never forgive, and retaliate to the deadline. 
More importantly, if players have any observation or action error, or imperfect rationality, the GT strategy players will soon go to a loss-loss situation and can never escape. This makes the social welfare become low. LR strategies can return back to mutual cooperation, even when there is an intentional or unintentional deviation by any player. 

Within each row in Figure \ref{fig:pd}, comparing with GT, the LR strategies show a clear advantage, which has been carefully discussed in the main text.
In each column, the performances of different LR strategies are different. Within the three candidate $k$, the LR strategy with $k=0.33$ performs the best under every error rate. It is worth noting that, given a certain error rate, one can search over the whole space of $k$, and generate different LR strategies and compute their payoffs in equilibrium. By numerically comparing the payoffs, a good $k$ against the given error rate can be obtained. 

Generally, higher error rate makes both grim trigger and LR perform worse. But LR strategies are more robust than the grim trigger strategy. Another observation is, as $k$ increases, the performance of LR strategy is more close to that of grim trigger strategy. This is intuitive, since an LR strategy with a full retaliation (i.e., $k=1$) simply degenerates to the grim trigger strategy, while an LR strategy with no retaliation (i.e., $k=0$) is the other extreme case where the player \emph{always} follows the collusive plan, even when there is a deviation. A larger $k$ makes the LR strategy more like a grim trigger.

\subsection{LR Strategies in Cournot Competition Game}
We also examine the LR strategies in the famous Cournot competition game. Cournot competition is a conventional economic model in which competing firms choose a quantity to produce independently and simultaneously.
The model applies when firms produce identical or standardized goods and it is assumed without additional institution or mechanism, they cannot collude or form a cartel.
The idea that one firm reacts to what it believes a rival will produce forms part of the perfect competition theory.

Consider two firms $i=1,2$ producing homogeneous products. Firms have market power, i.e., each firm's output decision affects the product's price. They compete in quantities, and choose quantities simultaneously, seeking to maximize profits.
Denote $q_i$ the quantity of firm $i$. The demand curve of the market is $$P=p_0-b(q_i+q_{-i}),$$
	where $p_0>c$ and $b>0$. The payoff function for firm $i$ is
	$$\pi_i(q_i,q_{-i})=\left[(p_0-b(q_i+q_{-i})-c\right]q_i.$$
	To find $i$'s optimal quantity, we can set the partial derivative as
	$$ \frac{\partial \pi_i}{\partial q_i}=p_0-2bq_i-bq_{-i}-c=0,  $$
	whereby the Nash quantity for this game is obtained as $q^N=\frac{p_0-c}{3b}$ and its Nash payoff is: $$\pi^N=(p_0-2bq^N-c)q^N=\frac{(p_0-c)^2}{9b}.$$
	
For the revision version of this Cournot competition, we still consider the symmetric equilibrium. Thus for any quantity $q$, the payoff for a symmetric quantity profile is $\pi(q)=(p_0-2bq-c)q$, and according to Definition 3, the deviation gain for firm-$1$ is
	\begin{equation}\begin{aligned}
	G(q)&=\max\limits_{q_1} \pi_1(q_1,q) -\pi_1(q,q)\\
	&=\max\limits_{q_1} [p_0-b(q_1+q)-c]q_1-(p_0-2bq-c)q\\
	&=\max\limits_{q_1} \{-bq_1^2+(p_0-bq-c)q_1+[(c-p_0)q+2bq^2]\}\\
	&=\frac{(p_0-c-3bq)^2}{4b}.
	\end{aligned}\end{equation}
	
	


With the above information, we can generate LR strategies. Then, with these obtained LR strategies, we can evaluate their performances in the Cournot competition revision game with error rates.
For simulations, we set parameters as $p_0=10,c=5,b=1$ and $\lambda=1$. For the LR strategies, we set $k$ to $0.35$, $0.5$ and $0.65$, respectively.

Two examples of the generated collusive MPC plans are shown in Figures \ref{Sfig:cd_performance}$(a)$ and \ref{Sfig:cd_performance}$(b)$. Due to the Cournot game structure, the collusive plan $x$ monotonically increases in time. The simulation results of LR strategies using these collusive plans and the responding $k$ are shown in Figures \ref{Sfig:cd_performance}$(c)$ to \ref{Sfig:cd_performance}$(h)$.
Each row is for simulations of an LR strategy with a fixed $k$. We can see that the performance of LR strategies in Cournot competition revision game is in a same pattern as that for Prisoner's Dilemma revision game.

\begin{figure*}[ht]
    \centering
    \subfigure[An MPC plan for with $k=0.35$.]
    {
        \includegraphics[scale=0.36]{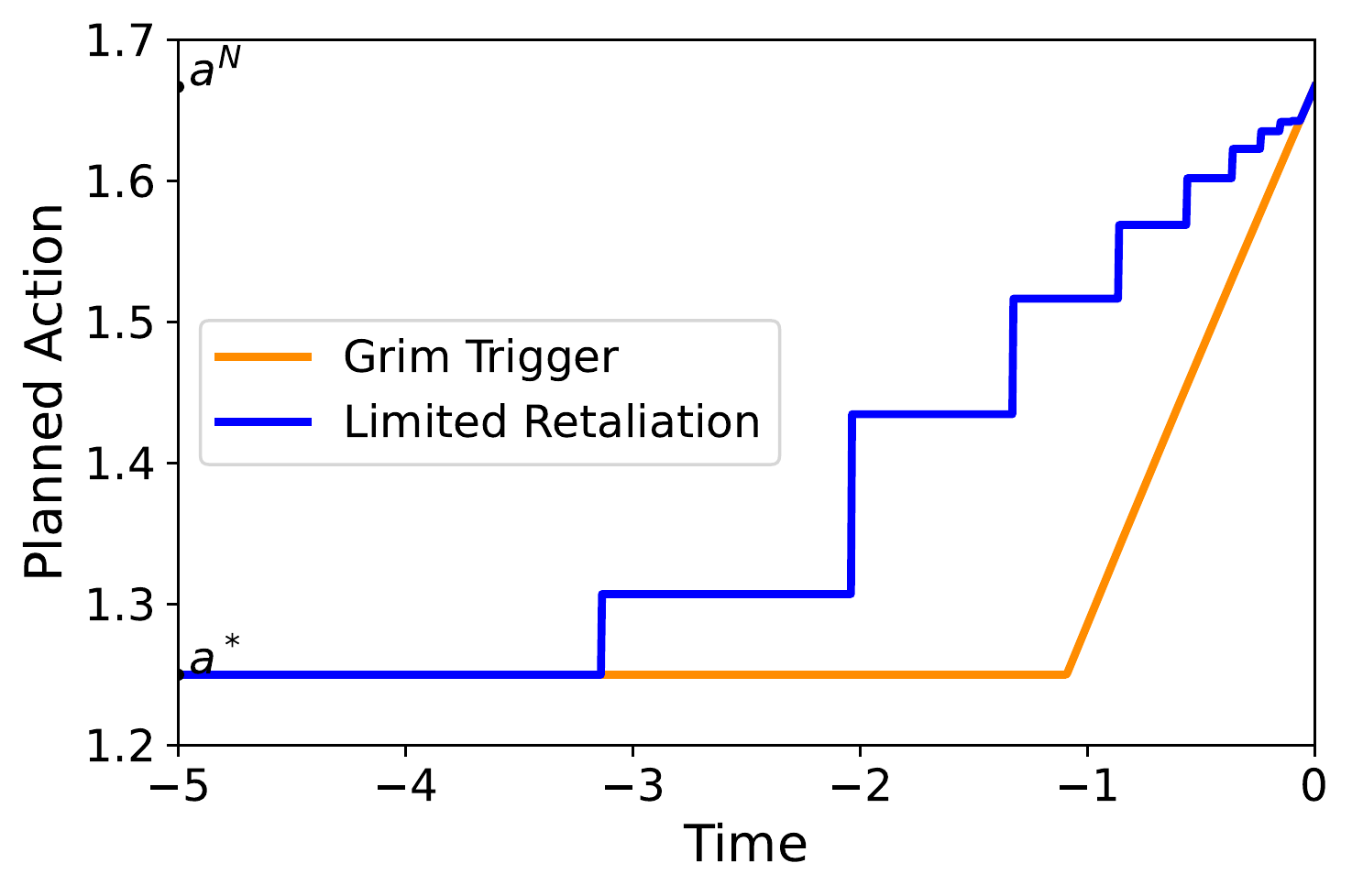}
        \label{fig:third_sub}
    }  
    \subfigure[An MPC plan for with $k=0.5$.]
    {
        \includegraphics[scale=0.36]{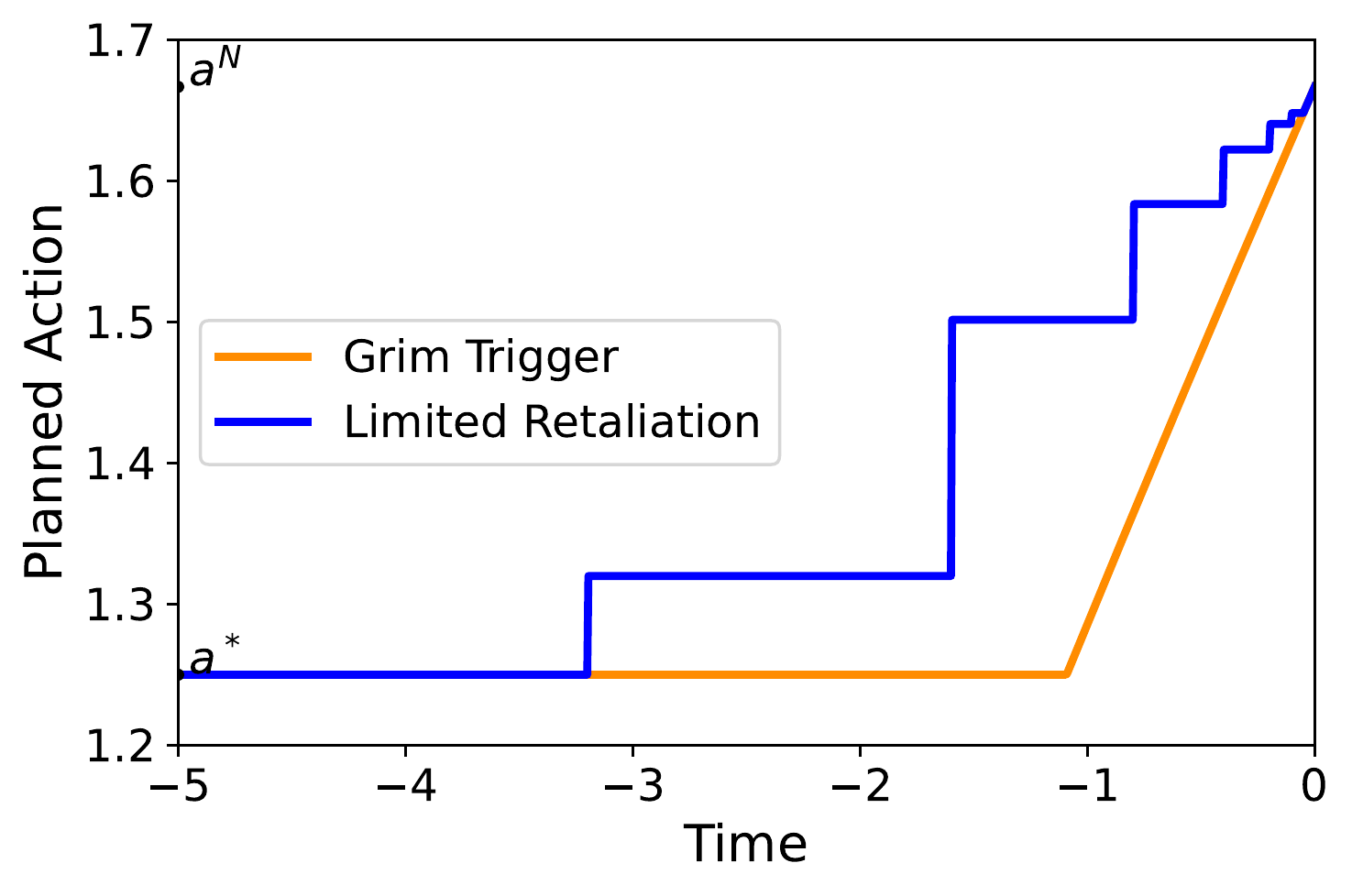}
        \label{fig:third_sub}
    } \\
    \subfigure[Payoff for $k=0.35$, $e=5\%$]
    {
        \includegraphics[scale=0.31]{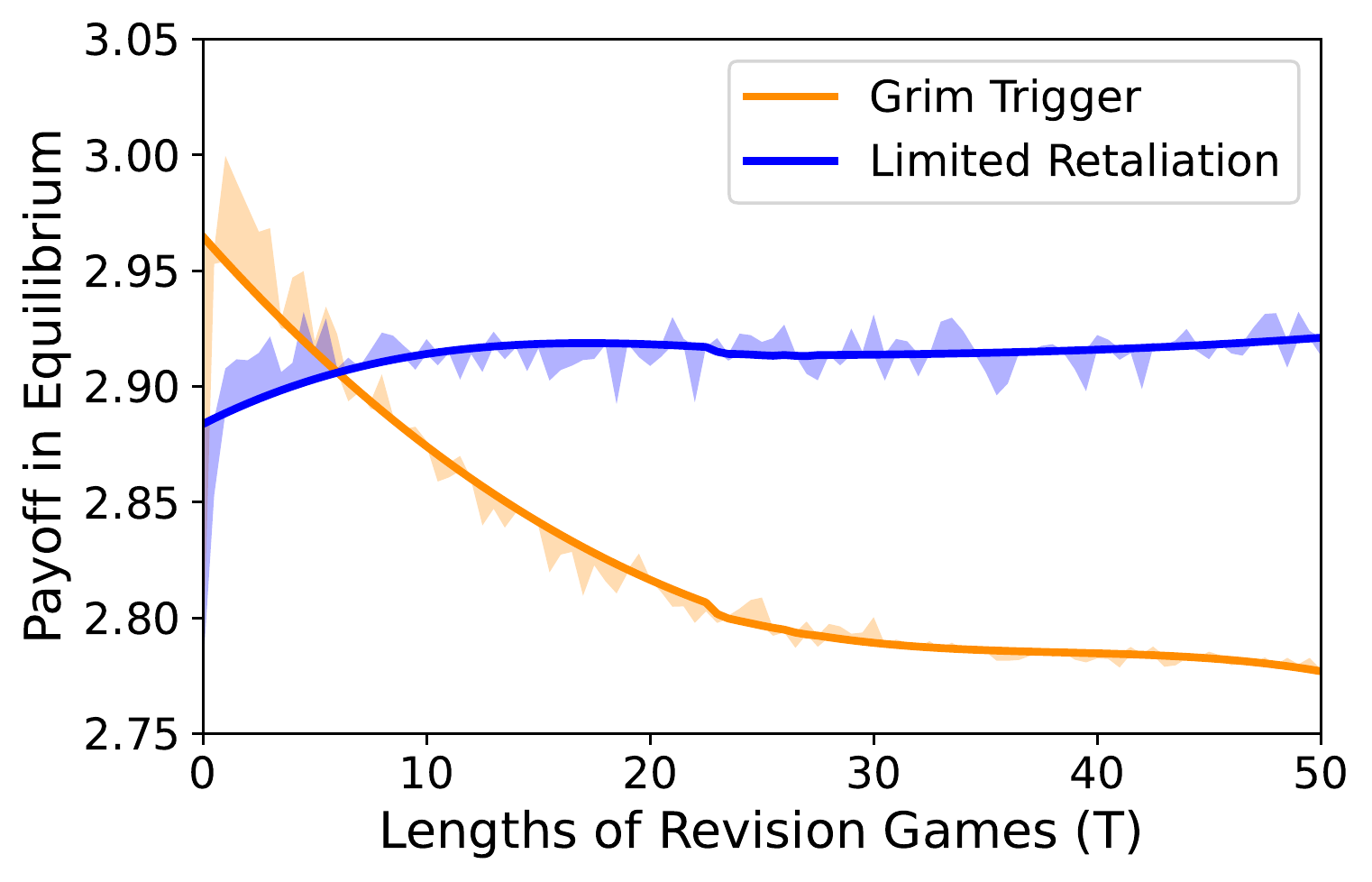}
        \label{fig:third_sub}
    }
    \subfigure[Payoff for $k=0.35$, $e=15\%$]
    {
        \includegraphics[scale=0.31]{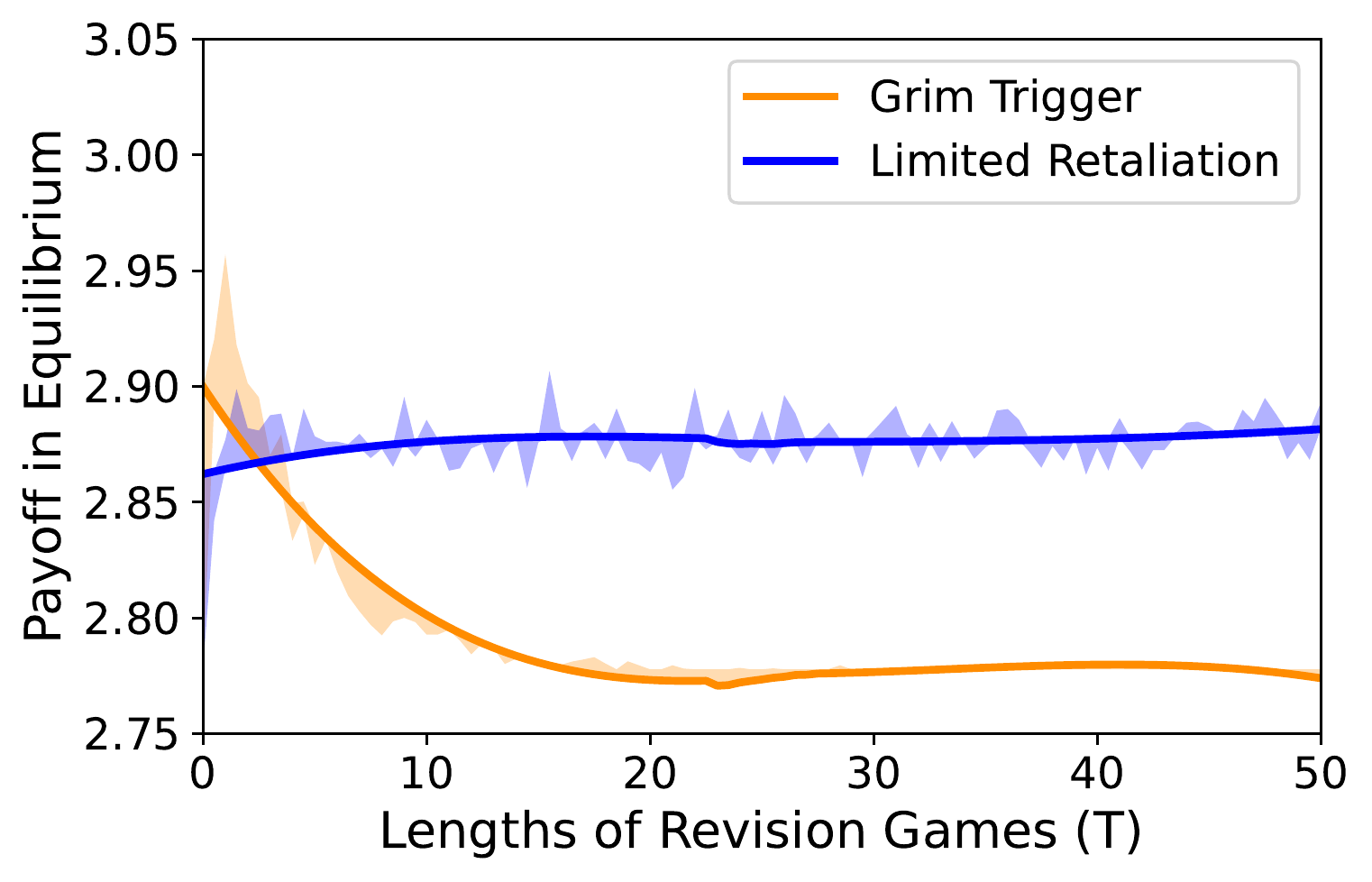}
        \label{fig:third_sub}
    }
        \subfigure[Payoff for $k=0.35$, $e=35\%$]
    {
        \includegraphics[scale=0.31]{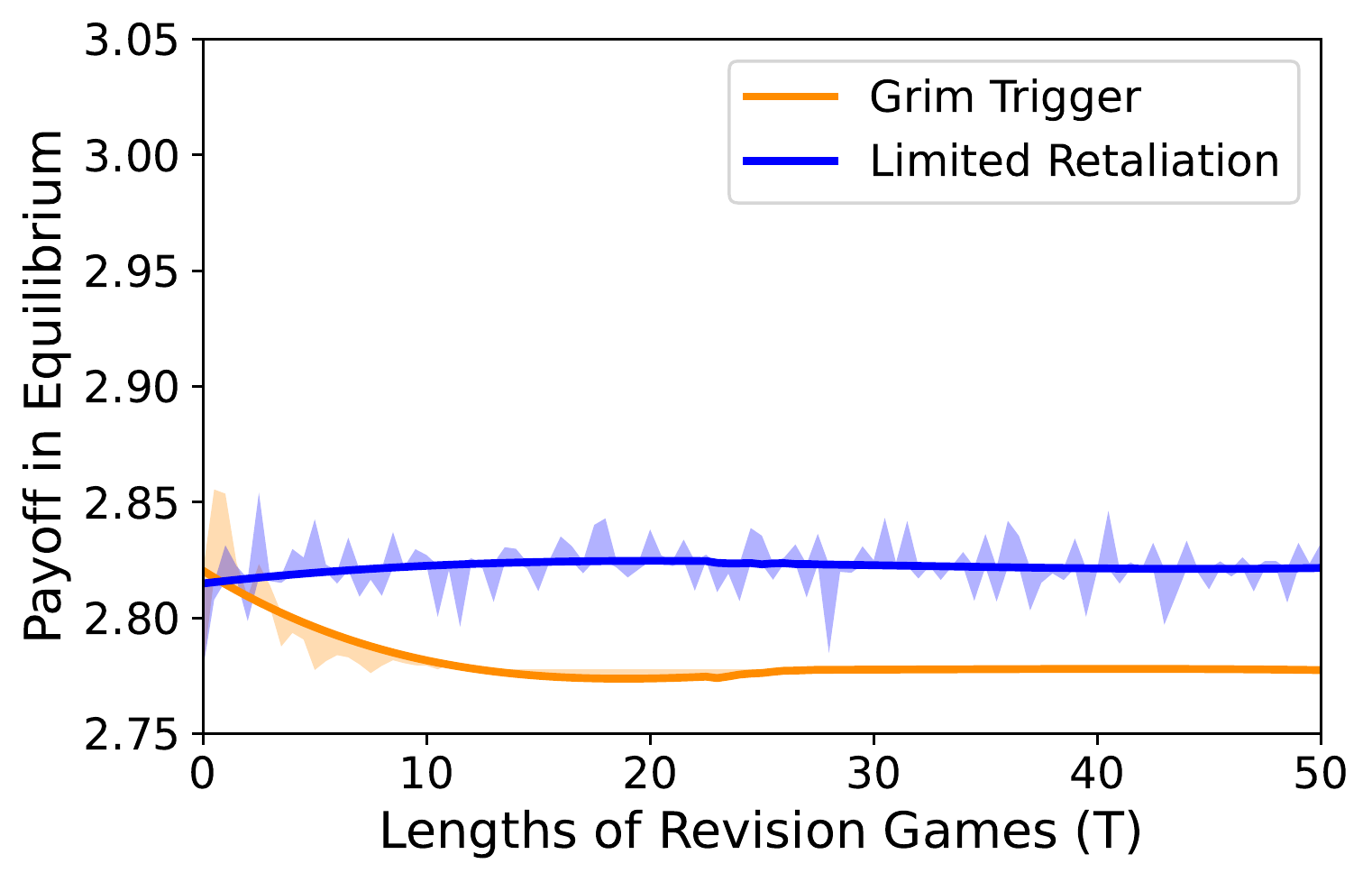}
        \label{fig:third_sub}
    }
            \subfigure[Payoff for $k=0.5$, $e=5\%$]
    {
        \includegraphics[scale=0.31]{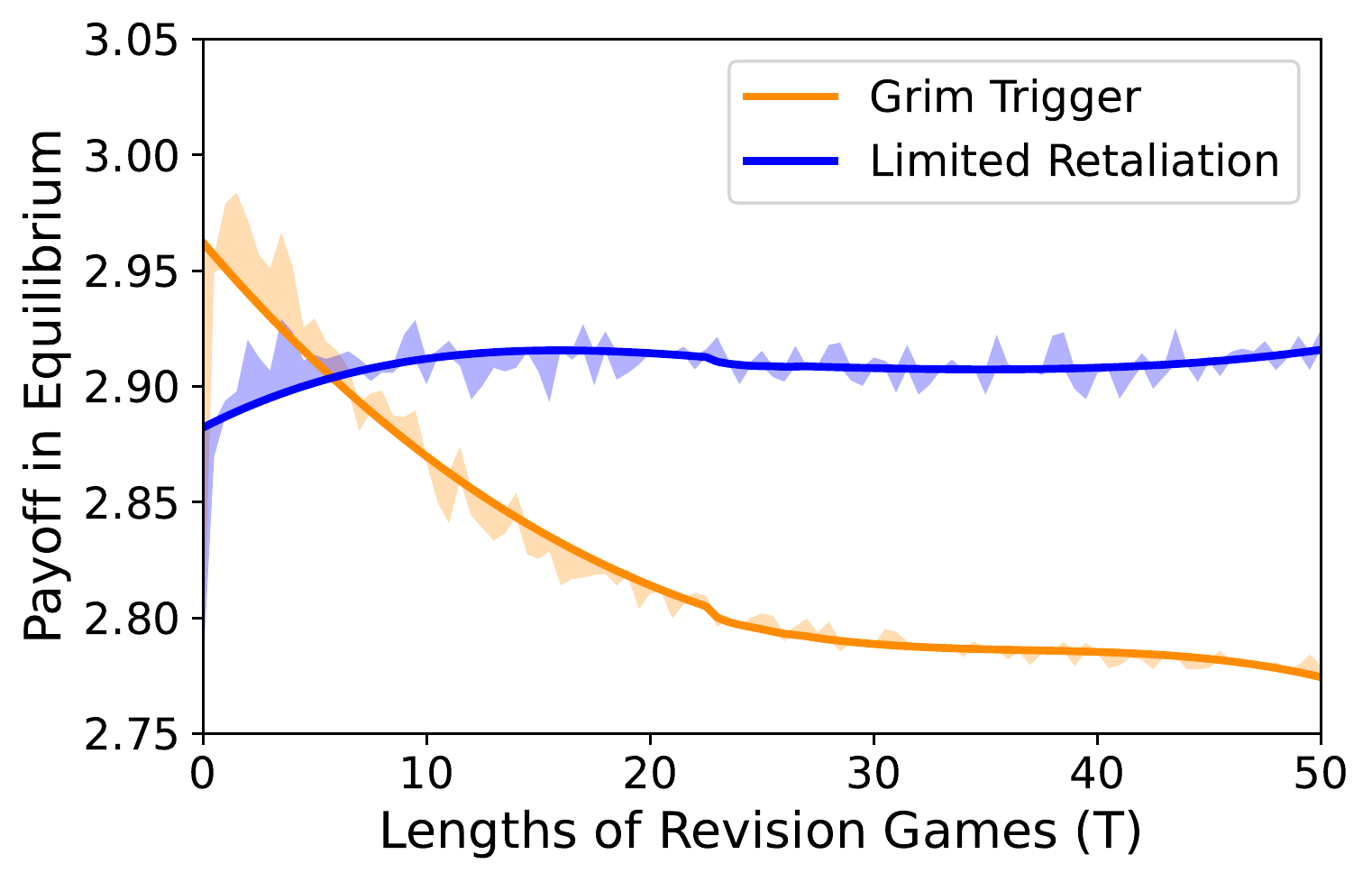}
        \label{fig:third_sub}
    }
        \subfigure[Payoff for $k=0.5$, $e=15\%$]
    {
        \includegraphics[scale=0.31]{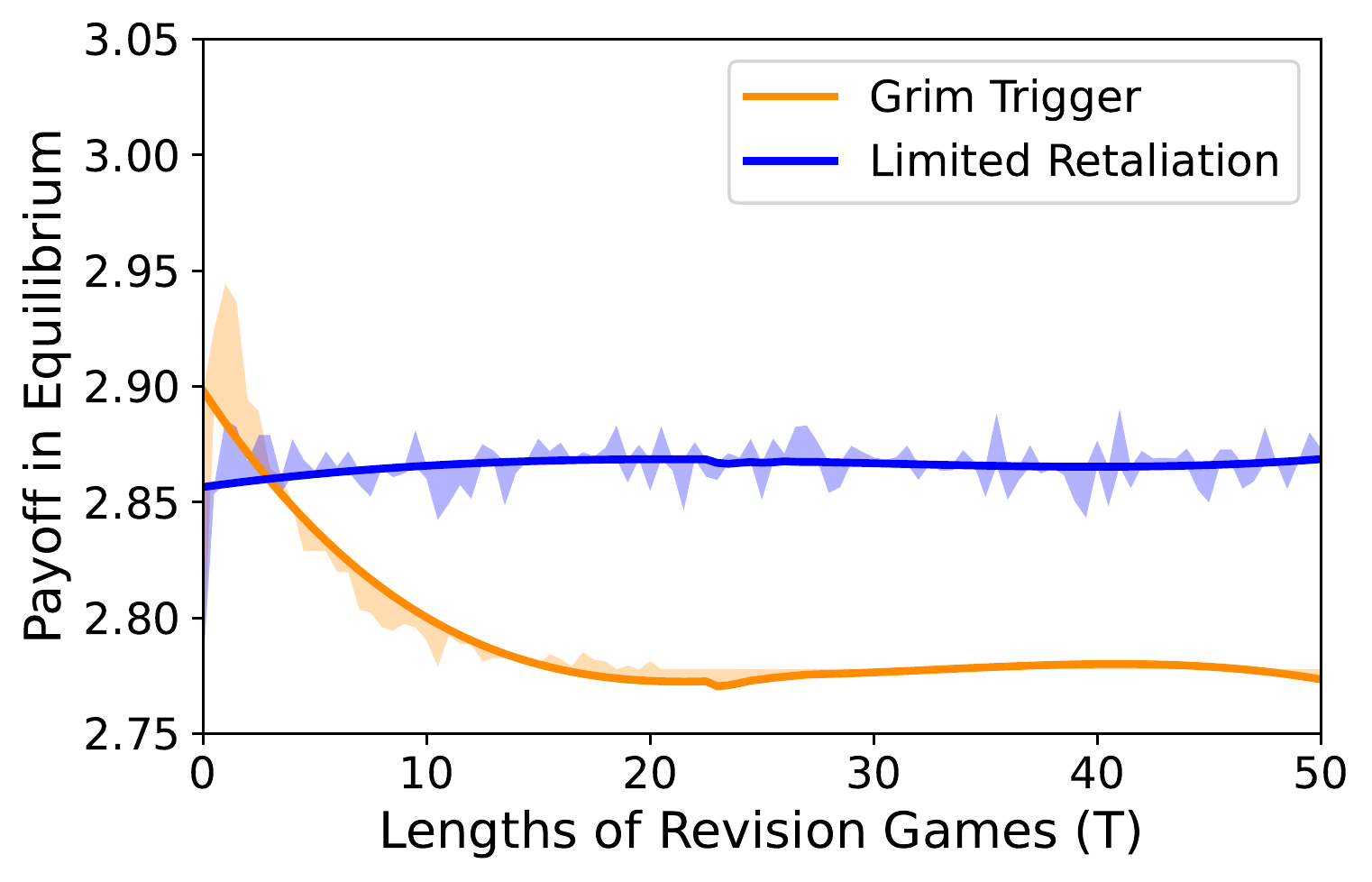}
        \label{fig:third_sub}
    }
    \subfigure[Payoff for $k=0.5$, $e=35\%$]
    {
        \includegraphics[scale=0.31]{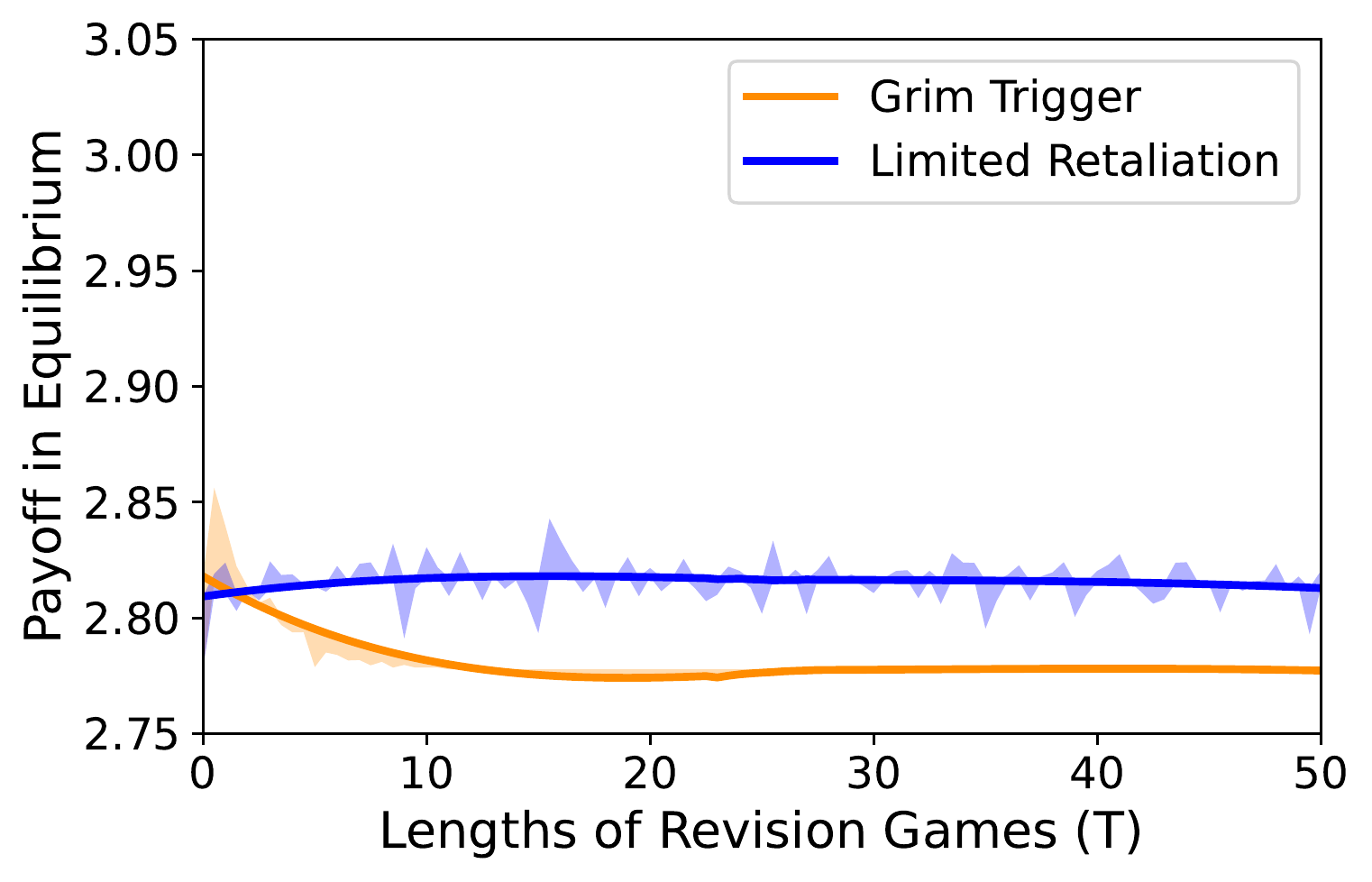}
        \label{fig:third_sub}
    }

    \caption{$(a)$ and $(b)$ show two Limited Retaliations strategies with different MPC plans for revision Cournot game, and $(c)$ to $(h)$ show their equilibrium payoffs under different error rates.}
    \label{Sfig:cd_performance}
\end{figure*}

	\section{Conclusions and Future Work}
    %
    This work is the first one in revision games capturing the mechanism of forgiveness. We analyze the incentive constraint for cooperation and forgiveness in revision games and identify the strategy of limited retaliation with retaliation duration limited to $k$ (LR strategies). LR strategies can return back to mutual cooperation, even when there is an intentional or unintentional deviation by any player. They are cooperative, vengeful, and also forgiving. These properties makes them especially suitable for real-world games where players confront intentional or unintentional errors. 
    We prove that LR strategies can sustain cooperation in subgame perfect equilibrium. We also design an algorithm to find the optimal piecewise constant plan for LR strategies. We finally implement this algorithm in both the continuous prisoner's dilemma and the Cournot revision games. 
    LR strategies is the first attempt to explain how vengeful agents should be and why forgiveness is important in multi-agent interactions with deadlines. They are robust and especially suitable for realistic revision games where players confront observation or action errors, or where players have bounded rationality.


    Before LR strategies, the existing Grim Trigger strategies are never forgiving, which excludes the possibility of resuming cooperation. These strategies are extremely fierce, which are not like realistic human behavior. More importantly, if players have any observation or action error, or imperfect rationality, the strategy players will soon go to a loss-loss situation and can never escape. This makes the social welfare become low. LR strategies can return back to mutual cooperation, even when there is an intentional or unintentional deviation by any player. They are cooperative, vengeful, and also forgiving. These properties make them especially suitable for real-world games where players confront observation or action errors. They well capture the mechanism of forgiveness, which has not yet been discussed in revision games. By the theory, grim trigger is better than LR in revision games if there is no error, this is because it provides the strongest threat. But the real world is noisy. As long as there is any action mistake, imperfect information or bounded rationality, a proper LR is better. To conquer the problems, famous strategies including Tit-for-Tat, Win-Stay Loss-Shift have been identified. But LR is the first strategy better than GT in noisy revision games. 
    
    Our immediate future work is to find other collusive plans and investigate their relationship with piecewise constant plans. This is because in the current work, we only consider that players both use a same LR strategy with a given $k$, which can constitute symmetric SPE. This is reasonable, because many game theory works consider symmetric strategy profile.
    The objective of our current paper is to propose a new equilibrium strategy solution of the revision game. The task about how to learn to, or how to evolve to, or how to coordinate to this equilibrium strategy is an important topic in the future. 
    Thus in the next step, we will relax the symmetric solution and discuss about how agents can know the strategy. We could use correlated equilibrium concept, commitment or learning schemes to investigate these issues.
    Moreover, in the current work, regarding ``one-shot deviation'', we concentrate on changes from the cooperative plan $x$, but do not consider deviation from the retaliation fierceness $k$. Because (1) technically, the investigation on function $x$ is already difficulty enough; and (2) the role of retaliation is to assist a good cooperative plan to be realized by the players' actual play and provide fine robustness against noise. The problem ``\textit{how to secure a cooperation}'' has more positive meaning than ``\textit{how to secure a retaliation}''.  Please note that this does not mean that our analysis is restricted to a specific $k$. The result in the main text can be used to check the performance of different $k$ for a given plan $x$; it can also provide the optimal $x$ for any given value of $k$.  
    

    Besides the above immediate future works, researches on counterfactual regret minimization \cite{zinkevich2007regret,celli2020no}, dynamic games cooperation \cite{martinez2015apology}, dynamic game distributed algorithms \cite{marden2009cooperative} could have underlying relationships with limited retaliation and forgiveness. The strategy adaptation in revision versions of Markov games \cite{zhan2020deep} is also a very promising topic. Regarding the solutions of subgame perfect equilibrium in our game, more advanced variations of backward induction \cite{li2013backward} is an important extension, which would be helpful for more realistic revision games. 
    Finally, applications of LR strategies in various realistic games, including networked games \cite{ye2014self}, real-time autonomous control \cite{chasparis2016design,spica2020real}, automated driving behavior \cite{lemmer2021maneuver}, and pollution regulation \cite{wang2022interplay} are also challenging but valuable topics, application of LR in various game mechanism design, especially in eBay like auction design, is also promising.

\bibliographystyle{unsrtnat}
\bibliography{aaai23} \label{end of main text}





\end{document}